\documentclass[draftcls,12pt,onecolumn]{IEEEtran}

\usepackage{times,amsmath,color,amssymb,graphicx,epsfig,cite,psfrag,enumerate,footnote}
\usepackage{subfigure,multirow,bm}
\usepackage[para]{threeparttable}
\newtheorem{Remark}{\it Remark}[section]
\newtheorem{Theorem}{\it Theorem}[section]
\newtheorem{Proposition}{\it Proposition}[section]
\newtheorem{Lemma}{\it Lemma}[section]

\begin{document}
%

\title{On the Alternative Relaying Diamond Channel with Conferencing Links}



\author{\IEEEauthorblockN{Chuan Huang,~\IEEEmembership{Student Member,~IEEE}, Shuguang
Cui},~\IEEEmembership{Member,~IEEE}

\thanks{Chuan Huang and Shuguang Cui are with the Department of Electrical
and Computer Engineering, Texas A\&M University, College Station,
TX, 77843. Emails: \{huangch, cui\}@tamu.edu.}}

\maketitle
\begin{abstract}
In this paper, the diamond relay channel is considered, which consists of one source-destination pair and two relay nodes connected with rate-limited out-of-band conferencing links. In particular, we focus on the half-duplex alternative relaying strategy, in which the two relays operate alternatively over time. With different amounts of delay, two conferencing strategies are proposed, each of which can be implemented by either a general two-side conferencing scheme (for which both of the two conferencing links are used) or a special-case one-side conferencing scheme (for which only one of the two conferencing links is used). Based on the most general two-side conferencing scheme, we derive the achievable rates by using the decode-and-forward (DF) and amplify-and-forward (AF) relaying schemes, and show that these rate maximization problems are convex. By further exploiting the properties of the optimal solutions, the simpler one-side conferencing is shown to be equally good as the two-side conferencing in term of the achievable rates under arbitrary channel conditions. Based on this, the DF rate in closed-form is obtained, and the principle to use which one of the two conferencing links for one-side conferencing is also established. Moreover, the DF scheme is shown to be capacity-achieving under certain conditions with even one-side conferencing. For the AF relaying scheme, one-side conferencing is shown to be sub-optimal in general. Finally, numerical results are provided to validate our analysis.
\end{abstract}




\section{Introduction}

Cooperative communication with relays is a well accepted concept to enhance system performance beyond-3G wireless technologies such as WiMAX and 3GPP UMTS Long Term Evolution (LTE). From the  information-theoretical viewpoint, the capacity bounds of the conventional three-node relay channel have been well studied \cite{cover,xie,kramer,madsen}, and various achievable schemes, such as decode-and-forward (DF) and compress-and-forward (CF), have been proposed. For the half-duplex relay channel, in \cite{madsen} and
the references therein the authors have studied the achievable rates and the associate power allocation schemes.

For the case with two relay nodes and no direct link between the source and the destination, termed as the diamond relay channel, various achievable rates were studied in \cite{schein,schein_thesis,niesen,xue,bagheri,resaei,chang}. In particular, the authors in \cite{schein,schein_thesis} discussed the capacity upper bound and the achievable rates using the DF and amplify-and-forward (AF) schemes under the full-duplex relaying mode. For the case with $N$ relays, the authors in \cite{niesen} used the bursty AF scheme to achieve the channel capacity within 1.8 bits for arbitrary channel gains and number of relays. Under the half-duplex mode, the authors in \cite{xue} discussed the achievable rates for two time-sharing schemes, i.e., the simultaneous relaying and alternative relaying. In \cite{bagheri}, the authors further discussed the same problem and bounded the gap between the achievable rate and the upper bound to be  within at most some constant number of bits. By further exploring partial collaboration between the two relays, the authors in \cite{resaei,chang} developed some DF schemes based on dirty paper coding (DPC) and block Markov encoding (BME), where the DF scheme is shown to be optimal in some special cases \cite{resaei}.

In practical cellular systems, some nodes might have extra out-of-band connections with the others, e.g., through Bluetooth, WiFi, optical fiber, etc., to exchange certain information and improve the overall system performance. From the information-theoretical viewpoint, such kind of interaction can be modeled as node conferencing \cite{cmac,ron,ron3,maric,Denus}. Specifically, the multiple access channel (MAC) \cite{cmac} with transmitter conferencing, the broadcast channel (BC) \cite{ron,ron3} with receiver conferencing, and the compound MAC (for which the two receivers try to decode both of the two source messages) with either transmitter or receiver conferencing \cite{maric,Denus}, have been studied.

In addition, the conferencing links can be used among relays to enhance the system performance \cite{chuan2}. In \cite{chuan2}, the authors considered the simultaneous relaying diamond channel with conferencing links, for which the two relays transmit and receive simultaneously over the same time and frequency. The achievable rates were studied for newly developed coding schemes based on the conventional DF, CF, and AF relaying schemes, and some capacity results were proved under special channel conditions. In this paper, we extend the idea of relay conferencing in \cite{chuan2} to the case with the alternative relaying strategy, for which only one relay transmits to the destination and the other listens the source in each time slot, and their working modes are exchanged in the consequent time slot. In practice, the alternative relaying strategy is more appealing than the simultaneous relaying strategy, due to the fact that it can achieve full multiplexing gain in the high signal-to-noise ratio (SNR) regime.

In this paper, we consider the following two different conferencing strategies:
\begin{enumerate}
  \item Conferencing strategy I: Relay conferencing for each source message is executed within the subsequent time slot after the relay receives the source signal by partially utilizing the conferencing links, and thus the decoding delay at the destination for each source message will be at most two time slots. By letting both or one of the two relays adopt the above conferencing strategy, we obtain the following two schemes:
      \begin{enumerate}
        \item Two-side conferencing: Use both of the two conferencing links, and both of the two relays are required to conference with each other;
        \item One-side conferencing: Use one of the conferencing links, and one of the relays sends message to the its counterpart, while the other one keeps silent.
      \end{enumerate}

  \item Conferencing strategy II: Relay conferencing for each source message is operated in the subsequent two time slots by fully utilizing the conferencing links, and thus the decoding delay for each source message at the destination will be more than two time slots. Similar to the previous case, we can also introduce both of the two-side and one-side conferencing schemes with this conferencing strategy.
\end{enumerate}
Intuitively, strategy II may lead to larger achievable rates compared to strategy I, since it allows higher conferencing rates between the two relays. Moreover, it is worth noting that one-side conferencing is just a special case of two-side conferencing, by letting one of the relays keep silent, and thus two-side conferencing in general will outperform one-side conferencing for both of the two conferencing strategies. Somewhat surprisingly, it will be shown that under certain conditions one-side conferencing is enough to achieve the same rate as the two-side conferencing, while it is much simpler to be implemented. The main contributions of this paper are summarized as follows: For both of the two conferencing strategies with the general two-side conferencing scheme, we derive the DF and AF achievable rates. For the DF relaying scheme, we formulate the rate maximization problem as a linear programming (LP) problem; for the AF relaying scheme, it is shown that the optimal linear combining problem is convex. By exploiting the properties of the optimal solutions for the above two problems, we further obtain the following results:
\begin{enumerate}
  \item For the DF relaying scheme, it is proved that the one-side conferencing scheme is optimal to achieve the maximum DF rates achieved by the two-side conferencing scheme for both of the two conferencing strategies. Based on this property, we derive the DF rates in closed-form under different channel coefficients, and further determine: (i) when relay conferencing is unnecessary; (ii) when relay conferencing is necessary, and which one of the conferencing links should be used. Moreover, we prove that the DF scheme achieves the capacity upper bound under conferencing strategy I asymptotically as the conferencing link rates go to infinity, while only finite conferencing rates are required under strategy II.
  \item For the AF relaying scheme, it is shown that: (i) When the second-hop relay-to-destination link SNRs become asymptotically large, two-side conferencing is necessary; (ii) when these link SNRs go to zero, one-side conferencing is asymptotically optimal, and each relay only needs to forward the signal with a higher SNR to the destination.
\end{enumerate}

The remainder of the paper is organized as follows. Section II introduces all assumptions and channel models. The capacity upper bound and achievable rates obtained by using the DF and AF relaying schemes are discussed for conferencing strategy I in Section III, and the DF scheme under conferencing strategy II is investigated in Section IV. Section V presents the numerical results. Finally, the paper is concluded in Section VI.

{\it Notation}: $\log(\cdot)$ stands for the base-2 logarithm; $\mathcal{C}(x) = \log \left( 1+ x\right)$ denotes for the additive white Gaussian noise (AWGN) channel capacity; $\min \left\{x,y \right\}$ and $\max \left\{x,y \right\}$ denote the minimum and maximum between two real numbers $x$ and $y$, respectively; $|x|$ denotes the amplitude of a complex number $x$, and $|\bm{A}|$ denote the determinant of a matrix $\bm{A}$; $\mathbb{E} (X)$ denotes the expectation of a random variable $X$.

\section{System Model}

We consider the diamond relay channel, as shown in Fig. \ref{system_model}, which consists of one source node, one destination node, and two relays. Between the two relays, there are two wired conferencing links, which are both rate-limited, i.e., the maximum supporting rate of the conferencing link from relay 1 to relay 2 is $C_{12} \geq 0 $, and $C_{21}$ is defined similarly. Due to the wired conferencing link assumptions, these two conferencing links can be considered orthogonal to each other and also orthogonal to the source-to-relay and relay-to-destination links. It is assumed that there are no direct wireless links between neither the source-destination pair nor the two relays.

In this paper, it is assumed that the transmissions of the source and relays are slotted, and the half-duplex alternative relaying scheme is adopted, as shown in Fig. \ref{alternative_relay}. Specifically, in the odd-numbered time slots, the source sends a message to relay 1, and relay 2 forwards a signal to the destination; in the even-numbered time slots, the roles of the two relays are exchanged. At the relays, the DF and AF relaying schemes are adopted for the transmissions to the destination. For the DF relaying scheme, we may allocate different time fractions to the odd and even time slots: Denote the time fraction allocated to the odd time slots as $\lambda_1$, and that to the even time slots as $\lambda_2$, with $\lambda_1 +\lambda_2 =1$, $\lambda_1 \geq 0$, and $\lambda_2 \geq 0$. Note that among odd time slots or among over time slots, they are of equal length. For the AF relaying scheme, we set $\lambda_1= \lambda_2 = \frac{1}{2}$.

The conferencing strategies shown in Fig. \ref{fig_scheduling} are described as follows. For conferencing strategy I, the source transmits independent messages $\{ w_{k} \}$ slot by slot; in the $k$-th time slot, when $k$ is odd, relay 1 listens to the source, and relay 2 sends two signals, one to the destination about the source messages $w_{k-2}$ and $w_{k-1}$, and the other to relay 1 about message $w_{k-1}$ via the wired conferencing links; when $k$ is even, the roles of these two relays are exchanged. For conferencing strategy II, in the $k$-th time slot, when $k$ is odd, relay 1 listens and relay 2 sends two signals, one to the destination about messages $w_{k-4}$ and $w_{k-1}$ within the $k$-th time slot, and the other to relay 1 about message $w_{k-1}$ spreading over both of the $k$-th and the $(k+1)$-th time slots; when $k$ is even, the roles of the two relays are exchanged. From Fig. \ref{fig_scheduling}, it is easy to see that strategy II fully utilizes the conferencing links, i.e., there are no idle time slots over the conferencing links, while strategy I only partially utilizes them. In Fig. \ref{fig_scheduling}, these two conferencing strategies for the most general two-side conferencing case are described. As stated in the introduction part, we are also interested in a special case of the two-side conferencing scheme, i.e., one-side conferencing. Taking conferencing strategy I as an example, we could just let relay 1 talk to relay 2 as in the two-side conferencing case, while relay 2 keeps silent; or only let relay 2 talk to relay 1 while relay 1 keeps silent. Similarly, the one-side conferencing scheme could as be defined for strategy II.

Note that for strategy II, the only way to deploy the conferencing strategy for the AF relaying scheme is to transmit the received signal at each relay to its counterpart repeatedly over the two subsequent conferencing time slots, since we can only forward the same received signal to the other relay via the conferencing links. As such, the two conferencing strategies for the AF relaying scheme are almost the same, only with different conferencing link SNRs (i.e., for conferencing strategy I, it suffers a $\frac{1}{2}$ penalty, with no penalty for strategy II.), and the results are quite similar (interestingly, we find that for the DF relaying scheme, it is not). Thus, in this paper, we only consider the AF relaying scheme under conferencing strategy I, and omit the analysis for the other strategy.

Due to relay conferencing, there may be extra decoding delay. To be concise, we describe the coding schemes by using $i$ and $(3-i)$, $i=1,2$, as the relay indices in the sequel. Take the $k$-th source message (to the $i$-th relay) for example: With conferencing strategy I, the destination needs to wait another $\lambda_i$-block\footnote{Here, one ``block'' consists of two successive time slots.} to obtain the signal from the $(3-i)$-th relay, which means that the decoding delay at the destination will be $\lambda_i$-block more compared with the case without relay conferencing\footnote{It is worth noting that even without relay conferencing, there is still an $\lambda_i$-block decoding delay due to the relaying operation.}. However, when transmitting $N$ messages in total, with $N$ going to infinity, the effect of decoding delay to the average achievable rate can be neglected. Accordingly, in this paper, we only consider the achievable rates over two successive time slots, since all coding schemes are operated periodically over time.

For the Gaussian channel case, the channel input-output relationships are given as follows. The received signal $y_i$ at the $i$-th relay from the source, $i=1,2$, is given as
\begin{eqnarray} \label{t1}
y_i= \sqrt{P_S} h_{i}x_i +n_i,~i=1,2,
\end{eqnarray}
where $x_i$ is the transmit signal from the source with unit average power, $P_S$ is the source transmit power, $h_{i}$ is the complex channel coefficient of the link from the source to the $i$-th relay, and $n_i$ is the independent and identically distributed (i.i.d.) circularly symmetric complex Gaussian (CSCG) noise with distribution $\mathcal{CN}(0, 1)$.

It is worth noting that in general we could allocate different power levels to the source messages at the odd and even time slots to maximize the overall system performance, which makes the achievable rate maximization problem hard to be tracked. In this paper, we focus more on the relay operations, i.e., conferencing and relaying; and thus we assume uniform power allocation at all the source messages for simplicity.

For the conferencing links, if the DF relaying scheme is adopted, we assume that the conferencing links are of limited capacities, below which the receivers can decode the transmitted messages without any errors. For the AF relaying scheme, we assume that the conferencing links are Gaussian channels, and the AF scheme is also adopted as the conferencing scheme, which will be discussed with more details later in Section III-C.

After relay conferecing, each relay generates a signal $t_i$ with unit average power, based upon the received signals from the source and the other relay. Then, the received signal $z_i$ at the destination from the $i$-th relay is given as
\begin{eqnarray} \label{t2}
z_i = \sqrt{P_R} g_{i} t_i + \widetilde{n}_i,~i=1,2,
\end{eqnarray}
where $P_R$ is the relay transmit power, $g_{i}$ is the complex channel coefficient of the link from the $i$-th relay to the destination, and $\widetilde{n}_i$ is the i.i.d. CSCG noise with distribution $\mathcal{CN} (0,1)$. For notation convenience, we denote the link SNRs as
\begin{align}
\gamma_i= |h_i|^2P_S,~ \widetilde{\gamma}_i=|g_i|^2P_R, ~i=1,2.
\end{align}

\section{Conferencing Strategy I}

In this section, we examine the capacity upper bound along with the DF and AF achievable rates for the considered channel with conferencing strategy I. Moreover, we prove some capacity achieving results under special channel conditions.

\subsection{Capacity Upper Bound}

In this subsection, we derive the capacity upper bound for the considered channel. Note that the following capacity upper bound is only applicable for the diamond relay channel with the alternative relaying strategy, not for the diamond relay channel with other strategies. To simplify notations, we call this bound as the capacity upper bound in this paper.

\begin{Theorem} \label{upper_bound}
\emph{Under conferencing strategy I, the capacity upper bound for the alternative relaying diamond channel with conferencing links is given as
\begin{align} \label{upper_bound_number}
C_{\text{upper}} &= \max_{\lambda_1 +\lambda_2 =1} \min \left\{  \lambda_1 \mathcal{C} (\gamma_1) +  \lambda_2 \mathcal{C} (\gamma_2), \lambda_2 \mathcal{C} (\widetilde{\gamma}_1) +  \lambda_1 \mathcal{C} (\widetilde{\gamma}_2), \lambda_1 C_{21} + \lambda_2 \left(  \mathcal{C} (\gamma_2) + \mathcal{C} (\widetilde{\gamma}_1) +C_{12}   \right), \right. \nonumber \\
&~~~~~~~~~~~~~~~~~~~~~~~~~~~~~~~~\left. \lambda_1 \left(  \mathcal{C} (\gamma_1) + \mathcal{C} (\widetilde{\gamma}_2) +C_{21}   \right) +\lambda_2 C_{12}  \right\}.
\end{align}}
\end{Theorem}
\begin{proof}
This bound is derived by the cut-set bound \cite{cover} considering the alternative relaying scheme with conferencing links as given in Section II. Similar analysis can be found in \cite{xue} and thus skipped.
\end{proof}

\subsection{DF Achievable Rate}

In this subsection, we first derive the DF rate for the general two-side conferencing, i.e., assuming that both of the two relays send information to each other via the conferencing links. After obtaining the most general expression of the DF rate in terms of a LP problem, we further exploit the properties of the optimal solution to simplify the coding scheme without sacrificing the DF rate, and show that one-side conferencing can also achieve the same DF rate as the two-side scheme. Then, we derive the DF rate in closed-form by solving the LP problem under different channel conditions.

\subsubsection{Rate Formulation}

First, we describe the main idea of the DF relaying scheme as follows: During the $k$-th time slot, the source transmits two messages $w^i_k$ and $w^{3-i}_k$ to the $i$-th relay by using superposition coding, with $w^i_k$ targeted at the destination via the relay-to-destination link and $w^{3-i}_k$ targeted at the $(3-i)$-th relay via the conferencing link; at the $(3-i)$-th relay, it transmits the messages $w^{3-i}_{k-1}$ and $w^{3-i}_{k-2}$ to the destination by using superposition coding, and $w^{i}_{k-1}$ to the $i$-th relay, respectively; at the end of the $k$-th time slot, the $i$-th relay decodes messages $w^i_k$, $w^{3-i}_k$, and $w^{i}_{k-1}$, and the destination decodes $w^{3-i}_{k-1}$ and $w^{3-i}_{k-2}$. Here, since all links in this channel are scheduled orthogonally over time or frequency, it is unnecessary to introduce any cooperation between the two relays and we only need to send independent messages cross the two time slots. Then, the DF rate is given in the following theorem.

\begin{Theorem}
\emph{Under conferencing strategy I, the DF achievable rate for the alternative relaying diamond channel with conferencing links is given as}
\begin{align}
\text{P1}:~& R_{\text{DF}}  =\max ~ R_{11} +R_{12} +R_{21} +R_{22} \label{DF_rate} \\
\text{s. t.}~~ & R_{3-i,i} \leq \lambda_i C_{3-i,i}, i=1,2, \nonumber \\
& R_{i,i} + R_{i,3-i} \leq \lambda_i \mathcal{C} (\gamma_i), i=1,2, \nonumber  \\
& R_{3-i,3-i} + R_{i,3-i} \leq \lambda_i \mathcal{C} (\widetilde{\gamma}_{3-i}), i=1,2, \nonumber \\
& \lambda_1 + \lambda_2 =1,R_{i,j} \geq 0,~i,j=1,2, \nonumber
\end{align}
\emph{where the design variables are $\left\{ R_{11},R_{12},R_{21},R_{22}, \lambda_1, \lambda_2 \right\}$, $R_{11}$ and $R_{12}$ are the rates of messages $w_{k}^{1}$ and $w_{k}^{2}$ decoded by relay 1, respectively, when $k$ is odd, and $R_{21}$ and $R_{22}$ are defined similarly when $k$ is even. }
\end{Theorem}
\begin{proof}
See Appendix \ref{appen_alternative_DF}.
\end{proof}

Note that if we do not send $w_k^{2}$ at odd $k$ and $w_k^{1}$ at even $k$, i.e., $R_{12}=R_{21}=0$, it is observed that the DF rate with conferencing given in (\ref{DF_rate}) is the same as that for the case without relay conferencing \cite{xue}, which implies that our coding scheme is a natural extension of that in \cite{xue}. Next, we show that one-side conferencing is enough to achieve the same maximum DF rate.

\begin{Proposition} \label{DF_one_side}
\emph{There exists one optimal point for Problem (P1) such that at least one of $R_{12}^*$ and $R_{21}^*$ is zero.}
\end{Proposition}
\begin{proof}
We prove this proposition by construction. Without loss of generality, assume that $R_{12}^* \geq  R_{21}^* > 0$, and $R_{i,j}^*,~i,j \in \{1,2\}$ is the optimal point of Problem (P1). Then, construct a new point as $\widehat{R}_{12} = R_{12}^*-  R_{21}^*$, $\widehat{R}_{21} =0$, and $\widehat{R}_{i,i}= R_{i,i}^* +  R_{21}^*$, $i=1,2$. It is easy to check that $\widehat{R}_{i,j}$'s also satisfy the constraints of Problem (P1) and achieve the same optimal value as $R_{i,j}^*$'s, $i,j \in \{ 1,2\}$. Thus, the proposition is proved.
\end{proof}

In practical communication system design, one-side conferencing simplifies the system requirements and thus is much easier to be implemented. In the next subsection, we will obtain the DF rate in closed-form, and show how to choose one of the two conferencing link for one-side conferencing under different channel conditions.

\subsubsection{Closed-form Expressions for the DF Rate}

With Proposition \ref{DF_one_side}, it is easy to check that the optimal value of Problem (P1) is equal to the maximum between those of the following two problems, which are recast from Problem (P1) by letting $R_{21}=0$ and $R_{12}=0$, respectively.
\begin{align}
\text{P1.1}:~& \max ~~R_{11} +R_{12} +R_{22} \label{DF_simp1} \\
\text{s. t.}~~ & R_{12} \leq \lambda_2 C_{12}, \label{DF_simp2} \\
& R_{11} + R_{12}  \leq \lambda_1 \mathcal{C} (\gamma_1) ,~R_{11}  \leq \lambda_2 \mathcal{C} (\widetilde{\gamma}_{1}) \label{DF_simp3}, \\
& R_{22} + R_{12} \leq \lambda_1 \mathcal{C} (\widetilde{\gamma}_{2}),~R_{22}  \leq \lambda_2 \mathcal{C} (\gamma_{2}) \label{DF_simp4}, \\
& R_{11} \geq 0, R_{12} \geq 0, R_{22} \geq 0,\lambda_1 + \lambda_2 =1, \lambda_i \geq 0, i=1,2, \label{DF_simp5}
\end{align}
and
\begin{align}
\text{P1.2}:~& \max ~~ R_{11} +R_{21} +R_{22} \label{DF_case2_simp1} \\
\text{s. t.}~~ & R_{21} \leq \lambda_1 C_{21}, \label{DF_case2_simp2} \\
& R_{11}  \leq \lambda_1 \mathcal{C} (\gamma_1),~R_{11}  + R_{21} \leq \lambda_2 \mathcal{C} (\widetilde{\gamma}_{1}) \label{DF_case2_simp3}, \\
& R_{22}  \leq \lambda_1 \mathcal{C} (\widetilde{\gamma}_{2}),~R_{22} + R_{21} \leq \lambda_2 \mathcal{C} (\gamma_{2})  \label{DF_case2_simp4}, \\
& R_{11} \geq 0, R_{12} \geq 0, R_{22} \geq 0, \lambda_1 + \lambda_2 =1, \lambda_i \geq 0, i=1,2. \label{DF_case2_simp5}
\end{align}
Note that if Problems (P1) and (P1.1) achieve the same optimal value, the optimal point $(R_{11}^*,R_{12}^*,R_{22}^*,\lambda_1^*, \lambda_2^*)$ of Problem (P1.1) is also feasible to Problem (P1) with $R_{21}^* = 0$, and such anoptimal point is also the solution for Problem (P1). Similar argument holds for the case that Problems (P1) and (P1.2) achieve the same optimal value. Therefore, solving Problems (P1.1) and (P1.2) is equivalent to solving Problem (P1). Before deriving the optimal solution for these two subproblems, we first introduce the following lemma.
\begin{Lemma}
\emph{There exists one optimal solution of Problem (P1.1), which makes either both of the two constraints in (\ref{DF_simp3}) or those in (\ref{DF_simp4}) satisfied with equality, i.e., $\lambda_1^* \mathcal{C}(\gamma_1) - R_{12}^* = \lambda_2^* \mathcal{C}(\widetilde{\gamma}_1)$ or $\lambda_1^* \mathcal{C}(\widetilde{\gamma}_2) - R_{12}^* = \lambda_2^* \mathcal{C}( \gamma_2)$. Similar result is true for Problem (P1.2).}
\end{Lemma}
\begin{proof}
We only prove the result for Problem (P1.1). If the optimal point of Problem (P1.1) satisfies that $R_{12}^* < \lambda_2^* C_{12}$, we can prove this lemma by contradiction using the same argument as that for Theorem 4.1 in \cite{xue}. Hence, we only need to consider the case with $R_{12}^* = \lambda_2^* C_{12}$. The main idea is shown as follows: By changing $\lambda_1^*$ to $\lambda_1^* + \epsilon$, where $\epsilon$ is a small real value, i.e., $|\epsilon| \ll 1$, it can be shown that all other cases cannot be optimal, except for the case shown in this lemma. There are four possible other cases.
\begin{enumerate}
  \item $\lambda_1^* \mathcal{C} (\gamma_{1}) - R_{12}^* < \lambda_2^* \mathcal{C} (\widetilde{\gamma}_{1}) $ and $\lambda_1^* \mathcal{C} (\widetilde{\gamma}_{2}) - R_{12}^* < \lambda_2^* \mathcal{C} (\widetilde{\gamma}_{2}) $: Choose $\epsilon > 0$, it is observed that the new solution satisfies $\widetilde{R}_{12} - R_{12}^* = - \epsilon C_{12}$, $\widetilde{R}_{11} - R_{11}^* =  \epsilon C_{12} + \epsilon \mathcal{C} (\gamma_{1})$, and $\widetilde{R}_{22} - R_{22}^* =  \epsilon C_{12} + \epsilon \mathcal{C} (\widetilde{\gamma}_{2})$. It is easy to check that the sum rate is improved, and thus this case cannot happen.
  \item $\lambda_1^* \mathcal{C} (\gamma_{1}) - R_{12}^* < \lambda_2^* \mathcal{C} (\widetilde{\gamma}_{1}) $ and $\lambda_1^* \mathcal{C} (\widetilde{\gamma}_{2}) - R_{12}^* > \lambda_2^* \mathcal{C} (\widetilde{\gamma}_{2}) $: Change $\lambda_1^*$ to $\lambda_1^* + \epsilon$, and it follows that the new solution satisfies $\widetilde{R}_{12} - R_{12}^* = - \epsilon C_{12}$, $\widetilde{R}_{11} - R_{11}^* =  \epsilon C_{12} + \epsilon \mathcal{C} (\gamma_{1})$, and $\widetilde{R}_{22} - R_{22}^* = - \epsilon \mathcal{C} (\gamma_{2})$. If $\gamma_1 \geq \gamma_2$, choose $\epsilon$ as a positive value; otherwise, choose $\epsilon$ as a negative value. It is easy to check that the sum rate is improved, and thus this case cannot happen.
  \item $\lambda_1^* \mathcal{C} (\gamma_{1}) - R_{12}^* > \lambda_2^* \mathcal{C} (\widetilde{\gamma}_{1}) $ and $\lambda_1^* \mathcal{C} (\widetilde{\gamma}_{2}) - R_{12}^* < \lambda_2^* \mathcal{C} (\widetilde{\gamma}_{2}) $: This case is similar to case 2).
  \item $\lambda_1^* \mathcal{C} (\gamma_{1}) - R_{12}^* > \lambda_2^* \mathcal{C} (\widetilde{\gamma}_{1}) $ and $\lambda_1^* \mathcal{C} (\widetilde{\gamma}_{2}) - R_{12}^* > \lambda_2^* \mathcal{C} (\widetilde{\gamma}_{2}) $: This case is similar to case 1).
\end{enumerate}
In conclusion, the proposition is proved.
\end{proof}

Next, we show how to obtain the optimal point $(R_{11}^*,R_{12}^*,R_{22}^*,\lambda_1^*, \lambda_2^*)$ of Problem (P1.1), where the solution of Problem (P1.2) can be obtained similarly.
\begin{enumerate}\renewcommand{\labelenumi}{(\roman{enumi})}
  \item If $\lambda_1^* \mathcal{C} (\gamma_1) -R_{12}^* = \lambda_2^* \mathcal{C} (\widetilde{\gamma}_1)$, it follows that $\lambda_1^* = \frac{\mathcal{C} (\widetilde{\gamma}_1) + R_{12}^* } {\mathcal{C} (\widetilde{\gamma}_1) + \mathcal{C} (\gamma_1)}$. Then, by (\ref{DF_simp2}) and noticing that $R_{12}^* \leq \lambda_1^* \mathcal{C} (\gamma_1)$ and $R_{12}^* \leq \lambda_1^* \mathcal{C} (\widetilde{\gamma}_2)$ (due to constraints (\ref{DF_simp3}) and (\ref{DF_simp4})), $R_{12}^*$ should satisfy the following conditions
\begin{align}
\left\{
  \begin{array}{l}
    R_{12}^* \leq  \frac{ \mathcal{C} \left( \gamma_1 \right) C_{12}} { \mathcal{C} \left( \gamma_1 \right) + \mathcal{C} \left( \widetilde{\gamma}_1 \right) + C_{12} } \\
    R_{12}^* \leq  \mathcal{C} \left( \gamma_1 \right) \\
    R_{12}^* \leq \frac{\mathcal{C} (\widetilde{\gamma}_1) \mathcal{C} ( \widetilde{\gamma}_2) }{\mathcal{C} (\widetilde{\gamma}_1) + \mathcal{C} (\gamma_1) -\mathcal{C} (\widetilde{\gamma}_2) }
  \end{array}
\right. .
\end{align}
Since the right-hand side of the second constraint is an upper bound of the right-hand side of the first one, it follows that the second one is redundant. Thus, we obtain
      \begin{align} \label{constraint_r12}
      0 \leq R_{12}^* \leq \min \left\{ \frac{\mathcal{C} ( \gamma_1) C_{12}}{\mathcal{C} ( \gamma_1) + \mathcal{C} ( \widetilde{\gamma}_1) +C_{12}} , \frac{\mathcal{C} (\widetilde{\gamma}_1) \mathcal{C} ( \widetilde{\gamma}_2) }{\mathcal{C} (\widetilde{\gamma}_1) + \mathcal{C} (\gamma_1) -\mathcal{C} (\widetilde{\gamma}_2) } \right\} \doteq k_1,
      \end{align}
      $R_{11}^* = \lambda_2^* \mathcal{C} (\widetilde{\gamma}_1)$, and $R_{22}^* = \min \left\{  \lambda_1^* \mathcal{C}(\widetilde{\gamma}_2) - R_{12}^*, \lambda_2^* \mathcal{C} (\gamma_2)  \right\}$. Thus, the optimal value of Problem (P1.1) is given as $R_{1} = \max_{(\ref{constraint_r12})} \left( R_{11}^* + R_{22}^* + R_{12}^* \right) = \max_{(\ref{constraint_r12})} r_1 (R_{12}^*)$, where
      \begin{align} \label{DF_casei}
      r_1 (R_{12}^*)  = \frac{1}{\mathcal{C} (\gamma_1) + \mathcal{C} (\widetilde{\gamma}_1) } \min
      \left\{ \begin{array}{c}
               \mathcal{C} (\gamma_1) \mathcal{C} (\widetilde{\gamma}_1) + \mathcal{C} (\gamma_1) \mathcal{C} (\gamma_2)  + \left( \mathcal{C} (\gamma_1) - \mathcal{C} (\gamma_2) \right) R_{12}^* \\
               \mathcal{C} (\gamma_1) \mathcal{C} (\widetilde{\gamma}_1) + \mathcal{C} (\widetilde{\gamma}_1) \mathcal{C} (\widetilde{\gamma}_2)  + \left( \mathcal{C} (\widetilde{\gamma}_2) - \mathcal{C} (\widetilde{\gamma}_1) \right) R_{12}^*
              \end{array}
        \right\} .
      \end{align}
  \item If $\lambda_1^* \mathcal{C}(\widetilde{\gamma}_2) - R_{12}^* = \lambda_2^* \mathcal{C}( \gamma_2)$, it follows that $\lambda_1^* = \frac{\mathcal{C} (\gamma_2) + R_{12}^* } {\mathcal{C} (\widetilde{\gamma}_2) + \mathcal{C} (\gamma_2)}$, and it is easy to check that
      \begin{align} \label{constraint_r21}
      0 \leq R_{12}^* \leq \min \left\{ \frac{\mathcal{C} ( \widetilde{\gamma}_2) C_{12}}{\mathcal{C} ( \gamma_2) + \mathcal{C} ( \widetilde{\gamma}_2) +C_{12}} , \frac{\mathcal{C} ( \widetilde{\gamma}_2) \mathcal{C} ( \gamma_1) }{\mathcal{C} ( \gamma_2) + \mathcal{C} ( \widetilde{\gamma}_2) - \mathcal{C} ( \gamma_1) } , \frac{ \left( \mathcal{C} ( \widetilde{\gamma}_2) \right)^2 }{\mathcal{C} ( \gamma_2)  }  \right\} \doteq k_2.
      \end{align}
      Thus, the optimal value of Problem (P1.1) is given as $R_2 = \max_{(\ref{constraint_r12})} r_2 (R_{12}^*)$, where
      \begin{align} \label{DF_caseii}
      r_2 (R_{12}^*) = \frac{1}{\mathcal{C} (\gamma_2) + \mathcal{C} (\widetilde{\gamma}_2) } \min
      \left\{ \begin{array}{c}
               \mathcal{C} (\gamma_2) \mathcal{C} (\widetilde{\gamma}_2) + \mathcal{C} (\gamma_1) \mathcal{C} (\gamma_2)  +  \left( \mathcal{C} (\gamma_1) - \mathcal{C} (\gamma_2) \right) R_{12}^* \\
               \mathcal{C} (\gamma_2) \mathcal{C} (\widetilde{\gamma}_2) + \mathcal{C} (\widetilde{\gamma}_1) \mathcal{C} (\widetilde{\gamma}_2)  + \left( \mathcal{C} (\widetilde{\gamma}_2) - \mathcal{C} (\widetilde{\gamma}_1) \right) R_{12}^*
              \end{array}
        \right\}.
      \end{align}
\end{enumerate}

It is worth noting that the two terms in the $\min$ operation of (\ref{DF_casei}) and (\ref{DF_caseii}) are all linear functions of $R_{12}^*$, and thus the optimal value of Problem (P1.1) is given by the max/min over two linear functions. Then, the optimal value of Problem (P1.1) can be obtained in the following cases, which are also shown in Fig. \ref{four_cases}.
\begin{enumerate}
  \item $\gamma_1 > \gamma_2$ and $\widetilde{\gamma}_2 > \widetilde{\gamma}_1$: As provable and shown in Fig. \ref{four_cases1}, both of the two functions in (\ref{DF_casei}) or (\ref{DF_caseii}) are strictly increasing over $R_{12}^*$, and thus the maximum values of (\ref{DF_casei}) and (\ref{DF_caseii}) are achieved at $R_{12}^* = k_1$ or $R_{12}^* = k_2$, respectively. Then, the optimal value of Problem (P1.1) is given as
      \begin{align} \label{DF_rate_case11}
      \max \left\{ r_1 (k_1), r_2(k_2)  \right\},
      \end{align}
      which implies that relay conferencing can strictly increase the DF rate in this case.
  \item $\gamma_1 \leq \gamma_2$ and $\widetilde{\gamma}_2 \leq \widetilde{\gamma}_1$: As provable and shown in Fig. \ref{four_cases2}, both of the two functions in (\ref{DF_casei}) or (\ref{DF_caseii}) are non-increasing over $R_{12}^*$, and thus the maximum values of (\ref{DF_casei}) and (\ref{DF_caseii}) are achieved at $R_{12}^* = 0$. Thus, the optimal value of Problem (P1.1) is given as
      \begin{align} \label{DF_rate_0}
        \max \left\{ r_1 (0), r_2(0)  \right\},
      \end{align}
      which implies that using the conferencing link from relay 1 to relay 2 cannot improve the DF rate for this case.
  \item $\gamma_1 > \gamma_2$ and $\widetilde{\gamma}_2 \leq \widetilde{\gamma}_1$:
      For either (\ref{DF_casei}) or (\ref{DF_caseii}), one function in the $\min$ operation is increasing over $R_{12}^*$, while the other one is non-increasing. As such, we need to further compare the constant terms in them, and there are two subcases:
      \begin{enumerate}
        \item $\mathcal{C} (\gamma_1) \mathcal{C} (\gamma_2) < \mathcal{C} (\widetilde{\gamma}_1) \mathcal{C} (\widetilde{\gamma}_2)$: As provable and shown in Fig. \ref{four_cases3}, for both (\ref{DF_casei}) and (\ref{DF_caseii}), the two functions in the $\min$ operation may have one intersection point for $R_{12}^* \geq 0$, which is given by $k_3 \doteq \frac{\mathcal{C} (\gamma_1) \mathcal{C} (\gamma_2) - \mathcal{C} (\widetilde{\gamma}_1) \mathcal{C} (\widetilde{\gamma}_2) }{\mathcal{C} (\widetilde{\gamma}_1) - \mathcal{C} (\widetilde{\gamma}_2) -\mathcal{C} (\gamma_1) + \mathcal{C} (\gamma_2)  } $. However, note that $k_3$ may not be within the region defined by (\ref{constraint_r12}) and (\ref{constraint_r21}). As such, for (\ref{DF_casei}) and (\ref{DF_caseii}), their maximum values are achieved at $k_{01} = \min \left( k_1 ,  k_3 \right)$ or $k_{02} = \min \left( k_2 ,  k_3 \right)$, respectively. Thus, the optimal value of Problem (P1.1) is given as
      \begin{align} \label{DF_inter}
      \max \left\{ r_1 (k_{01}), r_2(k_{02})  \right\},
      \end{align}
which means that relay conferencing can strictly increase the DF rate in this case.
              \item $\mathcal{C} (\gamma_1) \mathcal{C} (\gamma_2) \geq \mathcal{C} (\widetilde{\gamma}_1) \mathcal{C} (\widetilde{\gamma}_2)$: As provable and shown in Fig. \ref{four_cases4}, for both (\ref{DF_casei}) and (\ref{DF_caseii}), the two functions in the $\min$ operation have no intersection points in the region defined by (\ref{constraint_r12}) and (\ref{constraint_r21}), respectively. Thus, the optimal value of Problem (P1.1) is given the same as (\ref{DF_rate_0}), and it is concluded that by using the conferencing link from relay 1 to relay 2, the DF rate cannot be improved compared to the case without relay conferencing.
      \end{enumerate}
  \item $\gamma_1 \leq \gamma_2$ and $\widetilde{\gamma}_2 > \widetilde{\gamma}_1$: This case is similar to case 3), and the DF rate is given as
      \begin{enumerate}
      \item $\mathcal{C} (\gamma_1) \mathcal{C} (\gamma_2) > \mathcal{C} (\widetilde{\gamma}_1) \mathcal{C} (\widetilde{\gamma}_2)$: The optimal value of Problem (P1.1) is given by (\ref{DF_inter});
      \item $\mathcal{C} (\gamma_1) \mathcal{C} (\gamma_2) \leq \mathcal{C} (\widetilde{\gamma}_1) \mathcal{C} (\widetilde{\gamma}_2)$: The optimal value of Problem (P1.1) is given by (\ref{DF_rate_0}).
      \end{enumerate}
\end{enumerate}

Similar to Problem (P1.1), the optimal solution of Problem (P1.2) is summarized as follows.
\begin{enumerate}
  \item $\gamma_1 > \gamma_2$ and $\widetilde{\gamma}_2 > \widetilde{\gamma}_1$: The optimal value of Problem (P1.2) is given as
      \begin{align} \label{DF_p12_case1}
      \max \left\{ \widetilde{r}_1 (0), \widetilde{r}_2(0)  \right\},
      \end{align}
where $\widetilde{r}_1( R_{21}^*)$ and $\widetilde{r}_2( R_{21}^*)$ are defined as
\begin{align}
      \widetilde{r}_1 (R_{21}^*)  = \frac{1}{\mathcal{C} (\gamma_1) + \mathcal{C} (\widetilde{\gamma}_1) } \min
      \left\{ \begin{array}{c}
               \mathcal{C} (\gamma_1) \mathcal{C} (\widetilde{\gamma}_1) + \mathcal{C} (\widetilde{\gamma}_1) \mathcal{C} (\widetilde{\gamma}_2)  + \left( \mathcal{C} (\widetilde{\gamma}_1) - \mathcal{C} (\widetilde{\gamma}_2) \right) R_{21}^* \\
               \mathcal{C} (\gamma_1) \mathcal{C} (\widetilde{\gamma}_1) + \mathcal{C} (\gamma_1) \mathcal{C} (\gamma_2)  + \left( \mathcal{C} (\gamma_2) - \mathcal{C} (\gamma_1) \right) R_{21}^*
              \end{array}
        \right\} , \label{DF_caseiii}
      \end{align}
with
\begin{align*}
0 \leq R_{21}^* \leq \min \left\{ \frac{\mathcal{C} ( \widetilde{\gamma}_1) C_{21}}{\mathcal{C} ( \gamma_1) + \mathcal{C} ( \widetilde{\gamma}_1) +C_{21}} , \frac{\mathcal{C} (\widetilde{\gamma}_1^2) }{\mathcal{C} (2\widetilde{\gamma}_1) + \mathcal{C} (\gamma_1) } , \frac{\mathcal{C} ( \widetilde{\gamma}_1) \mathcal{C} ( \gamma_2)}{\mathcal{C} ( \gamma_1) + \mathcal{C} ( \widetilde{\gamma}_1) +\mathcal{C} ( \gamma_2)} \right\} \doteq \widetilde{k}_1,
\end{align*}
and
\begin{align}
\widetilde{r}_2 (R_{21}^*) & = \frac{1}{\mathcal{C} (\gamma_2) + \mathcal{C} (\widetilde{\gamma}_2) } \min
      \left\{ \begin{array}{c}
               \mathcal{C} (\gamma_2) \mathcal{C} (\widetilde{\gamma}_2) +  \mathcal{C} (\widetilde{\gamma}_1) \mathcal{C} (\widetilde{\gamma}_2)  +  \left( \mathcal{C} (\widetilde{\gamma}_1) - \mathcal{C} (\widetilde{\gamma}_2) \right) R_{21}^* \\
               \mathcal{C} (\gamma_2) \mathcal{C} (\widetilde{\gamma}_2) + \mathcal{C} (\gamma_1) \mathcal{C} (\gamma_2)  + \left( \mathcal{C} ( \gamma_2) - \mathcal{C} ( \gamma_1) \right) R_{21}^*
              \end{array}
        \right\}, \label{DF_caseiiii}
\end{align}
with
\begin{align*}
0 \leq R_{21}^* \leq \min \left\{\frac{\mathcal{C} ( \widetilde{\gamma}_2) C_{12}}{\mathcal{C} ( \gamma_2) + \mathcal{C} ( \widetilde{\gamma}_2) +C_{12}} , \frac{ \widetilde{\gamma}_1 \widetilde{\gamma}_2 }{ \widetilde{\gamma}_2 + \gamma_2 - \widetilde{\gamma}_1} \right\} \doteq \widetilde{k}_2.
\end{align*}

  \item $\gamma_1 \leq \gamma_2$ and $\widetilde{\gamma}_2 \leq \widetilde{\gamma}_1$: The optimal value of Problem (P1.2) is given as
      \begin{align} \label{DF_rate_case22}
      \max \left\{ \widetilde{r}_1 ( \widetilde{k}_1), \widetilde{r}_2( \widetilde{k}_2 )  \right\}.
      \end{align}
  \item $\gamma_1 > \gamma_2$ and $\widetilde{\gamma}_2 \leq \widetilde{\gamma}_1$:
      There are two possible subcases:
      \begin{enumerate}
        \item $\mathcal{C} (\gamma_1) \mathcal{C} (\gamma_2) < \mathcal{C} (\widetilde{\gamma}_1) \mathcal{C} (\widetilde{\gamma}_2)$: The optimal value of Problem (P1.2) is given as (\ref{DF_p12_case1}).

        \item $\mathcal{C} (\gamma_1) \mathcal{C} (\gamma_2) \geq \mathcal{C} (\widetilde{\gamma}_1) \mathcal{C} (\widetilde{\gamma}_2)$: The optimal value of Problem (P1.2) is given as
      \begin{align} \label{DF_rate_case21}
      \max \left\{ \widetilde{r}_1 ( \widetilde{k}_{01}), \widetilde{r}_2( \widetilde{k}_{02})  \right\},
      \end{align}
      where $\widetilde{k}_{01} = \min \left( \widetilde{k}_1 ,  \widetilde{k}_3 \right)$ and $ \widetilde{k}_{02} = \min \left( \widetilde{k}_2 ,  \widetilde{k}_3 \right)$, with $\widetilde{k}_3 \doteq \frac{\mathcal{C} (\gamma_1) \mathcal{C} (\gamma_2) - \mathcal{C} (\widetilde{\gamma}_1) \mathcal{C} (\widetilde{\gamma}_2) }{\mathcal{C} (\widetilde{\gamma}_1) - \mathcal{C} (\widetilde{\gamma}_2) -\mathcal{C} (\gamma_1) + \mathcal{C} (\gamma_2)  } $.

      \end{enumerate}
  \item $\gamma_1 \leq \gamma_2$ and $\widetilde{\gamma}_2 > \widetilde{\gamma}_1$: This is similar to case 3).
      \begin{enumerate}
      \item $\mathcal{C} (\gamma_1) \mathcal{C} (\gamma_2) > \mathcal{C} (\widetilde{\gamma}_1) \mathcal{C} (\widetilde{\gamma}_2)$: The optimal value of Problem (P1.2) is given by (\ref{DF_p12_case1});
      \item $\mathcal{C} (\gamma_1) \mathcal{C} (\gamma_2) \leq \mathcal{C} (\widetilde{\gamma}_1) \mathcal{C} (\widetilde{\gamma}_2)$: The optimal value of Problem (P1.2) is given by (\ref{DF_rate_case22}).
      \end{enumerate}
\end{enumerate}

From the above analysis, it is observed that under the same channel conditions, at most one between $R_{12}^*$ in Problem (P1.1) and $R_{21}^*$ in Problem (P1.2) can be non-zero. Thus, the optimal value of Problem (P1) is achieved by one of the Problems (P1.1) and (P1.2) with non-zero $R_{i,3-i}^*$ (if there is), since these constant terms in (\ref{DF_casei}) and (\ref{DF_caseiii}) (same for (\ref{DF_caseii}) and (\ref{DF_caseiiii})) are identical; for the case that both of $R_{i,3-i}^*$'s are zero, Problems (P1.1) and (P1.2) render the same optimal value, which is the same as that of Problem (P1). Therefore, we could obtain the DF rate in closed-form under different channel conditions, which is summarized in Table \ref{which_conf_link}. Moreover, it is worth noting that under arbitrary channel conditions, at most one between $R_{12}^*$ and $R_{21}^*$ is positive, which is coherent with the result in Proposition \ref{DF_one_side} and indicates which one of the conferencing links should be used; furthermore, for some cases, both $R_{12}^*$ and $R_{21}^*$ are zero, which means under these channel conditions, relay conferencing is useless. In Table \ref{which_conf_link}, we also summarize which conferencing link should be used to deploy one-side relay conferencing and when relay conferencing cannot improve the DF rate.

\begin{Remark}
From the above analysis, we observe that for the symmetric channel case, i.e., $\gamma_1 = \gamma_2$ and $\widetilde{\gamma}_1 = \widetilde{\gamma}_2$, relay conferencing cannot improve the DF rate with the alternative relaying scheme, which is not true for the simultaneous relaying scheme \cite{chuan2}.
\end{Remark}

\subsubsection{Asymptotic Performance}

From \cite{xue}, we know that for the diamond relay channel without relay conferencing, the DF scheme achieves the capacity upper bound under arbitrary channel conditions. In the following, we show an asymptotic capacity-achieving result for the considered channel in this paper.
\begin{Proposition} \label{DF_schemeI_asym}
\emph{With conferencing strategy I and arbitrary channel coefficients, the DF relaying scheme achieves the capacity upper bound given in (\ref{upper_bound_number}) asymptotically as the conferencing link rates go to infinity.}
\end{Proposition}
\begin{proof}
When $C_{i,3-i}$'s go to infinity, it is easy to see that the capacity upper bound given in (\ref{upper_bound_number}) is asymptotically equal to
\begin{align}
C_{\text{upper}}^{\infty} = \max_{\lambda_i} \min \left\{  \lambda_1 \mathcal{C} (\gamma_1) +  \lambda_2 \mathcal{C} (\gamma_2), \lambda_2 \mathcal{C} (\widetilde{\gamma}_1) +  \lambda_1 \mathcal{C} (\widetilde{\gamma}_2)  \right\}.
\end{align}
On the other hand, we set $R_{12} =R_{21} =0$ in (\ref{DF_rate}), and we observe that $R_{\text{DF}} \geq C_{\text{upper}}^{\infty} $. Thus, the proposition is proved.
\end{proof}

\begin{Remark} \label{DF_schemeI_no_capa}
For finite and positive $C_{i,3-i}$'s, it can be shown that the third and the fourth terms in (\ref{upper_bound_number}) are generally larger than the DF rate define in (\ref{DF_rate}), which implies why the DF relaying scheme cannot achieve this capacity upper bound under general channel conditions. To see this point, we fix $\lambda_i$, $i=1,2$, and sum the following three constraints in (\ref{DF_rate}) together: $R_{12} \leq \lambda_2 C_{12}$, $R_{22} + R_{21} \leq \lambda_2 \mathcal{C} (\gamma_2)$, and $R_{11} + R_{21} \leq \lambda_2 \mathcal{C} (\widetilde{\gamma}_1)$, which leads to
\begin{align}
R_{\text{DF}} & \leq R_{11} + R_{12} + 2 R_{21} +R_{22} \label{DF_finite_conf1}\\
& = \lambda_2  \left( \mathcal{C} (\gamma_2) + \mathcal{C} (\widetilde{\gamma}_1) + C_{12} \right) \\
&\leq \lambda_1 C_{21} + \lambda_2   \left( \mathcal{C} (\gamma_2) + \mathcal{C} (\widetilde{\gamma}_1) + C_{12} \right), \label{DF_finite_conf3}
\end{align}
where these two equalities in (\ref{DF_finite_conf1}) and (\ref{DF_finite_conf3}) are achieved only when $R_{21} = C_{21} = 0$. In general, since $C_{i,3-i} > 0$, we conclude that the capacity upper bound given in (\ref{upper_bound_number}) cannot be achieved by the DF scheme.
\end{Remark}

\subsection{AF Achievable Rate}

For the AF relaying scheme, each relay first linearly combines the received signals from the source and the other relay, and then transmits the combination to the destination under an individual relay power constraint.

We assume that the conferencing links are Gaussian. For simplicity, let the input of the conferencing link at the $(3-i)$-th relay be $x_{3-i,i} = y_{3-i}$ and the link gain of each conferencing link equal to 1. Thus, the conferencing link output at the $i$-th relay is given as
\begin{align}
y_{3-i,i} = x_{3-i,i} + n_{3-i,i},
\end{align}
where $n_{3-i,i}$ is the i.i.d. CSCG noise at the $i$-th relay with a distribution $\mathcal{CN}\left( 0, \sigma_{3-i,i}^2 \right)$. With the conferencing link rate constraint, $\sigma_{3-i,i}^2$ is given as
\begin{align} \label{AF_conf_noise}
\sigma^2_{3-i,i} = \frac{ \gamma_{3-i} +1}{2^{C_{3-i,i}/2}-1}, i=1,2.
\end{align}

After the relay conferencing, each relay combines the two received signals from the source and the other relay as $t_i = a_{ii}y_i + a_{3-i,i} y_{3-i,i}$, where $a_{ii}$ and $a_{3-i,i}$ are the complex combining parameters, satisfying the following average transmit power constraint
\begin{align} \label{AF_power_constraints}
  \mathbb{E} \left( |t_i|^2 \right)&  = |a_{ii}|^2 \left( \gamma_i +1 \right) +|a_{3-i,i}|^2 \left( \gamma_{3-i} + 1 +\sigma^2_{3-i,i} \right) \leq 1,~i=1,2.
\end{align}

At the destination, we apply a sequential decoding process: Assume that at the $k$-th time slot, the previous $k-1$ source messages have already been successfully decoded; decode the $k$-th source message based on the received signals at the $k$-th and the $(k+1)$-th time slots, by treating the $(k+1)$-th message at the $(k+1)$-th time slot as noise. As such, the AF rate for each source message is given as
\begin{align} \label{AF_rate_original}
R_{i} & = \frac{1}{2} \mathcal{C} \left( \frac{ |a_{ii}|^2 \gamma_i \widetilde{\gamma}_i}{1  +  |a_{ii}|^2 \widetilde{\gamma}_i}  +  \frac{|a_{i,3-i}|^2 \gamma_i \widetilde{\gamma}_{3-i} }{ |a_{i,3-i}|^2 \widetilde{\gamma}_{3-i} \left( 1 + \sigma_{i,3-i}^2 \right) + |a_{3-i,3-i}|^2  \widetilde{\gamma}_{3-i} (\gamma_{3-i}+ 1) +1 }  \right),
\end{align}
where $R_1$ and $R_2$ denote the rates for the source messages in odd and even time slots, respectively. Thus, the AF rate is given as
\begin{align} \label{AF_rate}
R_{\text{AF}} = \max_{(\ref{AF_power_constraints})}R_1 + R_2.
\end{align}
Then, we have the following result for the convexity of Problem (\ref{AF_rate}).
\begin{Proposition} \label{af_convex_prop}
The AF rate maximization problem in (\ref{AF_rate}) is concave over the combining parameters $|a_{i,j}|^2$, $i,j \in \{1,2\}$.
\end{Proposition}
\begin{proof}
See Appendix \ref{convex_af}.
\end{proof}

Even though the AF rate maximization problem in (\ref{AF_rate}) can be solved by numerical algorithms, e.g., the interior point method \cite{boyd}, we know little about whether two-side conferencing is necessary with general channel coefficients. To obtain some insights for the proposed conferencing scheme, we further investigate the performance of the AF scheme for the cases when the second-hop link SNR $\widetilde{\gamma}_i$ goes to infinity and zero, respectively.

\subsubsection{High SNR regime} For the case with $\widetilde{\gamma}_i \rightarrow \infty$, (\ref{AF_rate_original}) can be approximated as
\begin{align} \label{AF_rate_high_SNR}
R_{i} \approx \frac{1}{2} \mathcal{C} \left(  \gamma_i  +  \frac{|a_{i,3-i}|^2 \gamma_i}{ |a_{i,3-i}|^2  \left( 1 + \sigma_{i,3-i}^2 \right)+  |a_{3-i,3-i}|^2 \gamma_{3-i}}  \right).
\end{align}

\begin{Remark}
In general, it is still difficult to derive the closed-form solution for Problem (\ref{AF_rate}) with (\ref{AF_rate_high_SNR}). However, it is worth noting that the term $\gamma_i$ in $\mathcal{C}(\cdot)$ of (\ref{AF_rate_high_SNR}) equals the received SNR for the case without relay conferencing when $\widetilde{\gamma}_i \rightarrow \infty$; moreover, the second term in $\mathcal{C}(\cdot)$ of (\ref{AF_rate_high_SNR}) is the gain from relay conferencing. For some special cases, i.e., where $\gamma_1 = \gamma_2$ and we choose $C_{i,3-i}$ such that $1 + \sigma_{i,3-i}^2 = \gamma_i$, it can be checked that the maximum AF rate is achieved at $|a_{i,j}|^2 = \frac{1}{2\gamma_1}$, $i,j \in \{ 1,2\}$ (note that $R_i$ is concave over $|a_{i,j}|^2$'s). Thus, for both $R_1$ and $R_2$, the conferencing gains are non-zero, which implies that two-side conferencing is necessary.
\end{Remark}

\subsubsection{Low SNR regime} For the case with $\widetilde{\gamma}_i \rightarrow 0$, (\ref{AF_rate_original}) can be approximated as
\begin{align}
R_{i} \approx \frac{1}{2} \left( |a_{ii}|^2 \gamma_i \widetilde{\gamma}_i  +  |a_{i,3-i}|^2 \gamma_i \widetilde{\gamma}_{3-i}  \right).
\end{align}
Then, the AF rate maximization problem can be recast as
\begin{align}
(\text{P}2)~\max_{(\ref{AF_power_constraints})} ~~\frac{1}{2} \left( |a_{11}|^2 \gamma_1 \widetilde{\gamma}_1  +  |a_{12}|^2 \gamma_1 \widetilde{\gamma}_{2} + |a_{22}|^2 \gamma_2 \widetilde{\gamma}_2  +  |a_{2,1}|^2 \gamma_2 \widetilde{\gamma}_{1}  \right),
\end{align}
which is a LP problem. It can be shown that Problem (P2) can be decomposed into two subproblems, and its optimal point can be constructed from those of the following two problems.
\begin{align}
(\text{P}2.1)~& \max ~~\frac{1}{2} \left( |a_{11}|^2 \gamma_1 \widetilde{\gamma}_1  +  |a_{2,1}|^2 \gamma_2 \widetilde{\gamma}_{1}  \right) \\
\text{s. t.}~ &~~ |a_{11}|^2 \left( \gamma_1 +1 \right) +|a_{21}|^2 \left( \gamma_{2} + 1 +\sigma^2_{21} \right) \leq 1,
\end{align}
and
\begin{align}
(\text{P}2.2)~& \max ~~\frac{1}{2} \left(   |a_{12}|^2 \gamma_1 \widetilde{\gamma}_{2} + |a_{22}|^2 \gamma_2 \widetilde{\gamma}_2    \right) \\
\text{s. t.} ~&~~ |a_{22}|^2 \left( \gamma_2 +1 \right) +|a_{12}|^2 \left( \gamma_{1} + 1 +\sigma^2_{12} \right) \leq 1.
\end{align}
It is easy to check that one of the following two points is optimal for Problem (P2.1): $\left( |a_{11}|^2, |a_{21}|^2 \right) = \left( \frac{1}{ \gamma_1 +1}, 0  \right)$ and $\left( |a_{11}|^2, |a_{21}|^2 \right) = \left(  0 , \frac{1}{\gamma_{2} + 1 +\sigma^2_{21} } \right)$, and thus its optimal value is $\widetilde{\gamma}_1 \cdot \max \left\{ \frac{\gamma_1}{\gamma_1 + 1}, \frac{\gamma_2}{\gamma_2 + 1 +\sigma_{21}^2} \right\}$. Similarly, for Problem (P2.2), its optimal value $\widetilde{\gamma}_2 \cdot  \max \left\{ \frac{\gamma_2}{\gamma_2 + 1}, \frac{\gamma_1}{\gamma_1 + 1 +\sigma_{12}^2} \right\}$ is achieved by either $\left( |a_{22}|^2, |a_{12}|^2 \right) = \left( \frac{1}{ \gamma_2 +1}, 0  \right)$ or $\left( |a_{22}|^2, |a_{12}|^2 \right) = \left(  0 , \frac{1}{\gamma_{1} + 1 +\sigma^2_{12} } \right)$. By considering different combinations of the possible optimal points of Problems (P2.1) and (P2.2), we obtain the optimal solution $|a_{i,j}^*|^2$, $i,j \in \{ 1,2\}$, of Problem (P2) as follows.
\begin{enumerate}
  \item For the case with $\frac{\gamma_1}{\gamma_1 + 1} \geq \frac{\gamma_2}{\gamma_2 + 1 +\sigma_{21}^2}$ and $ \frac{\gamma_2}{\gamma_2 + 1} \geq \frac{\gamma_1}{\gamma_1 + 1 +\sigma_{12}^2} $, we have $|a_{11}^*|^2= \frac{1}{\gamma_1 + 1}$, $|a_{22}^*|^2= \frac{1}{\gamma_2 + 1}$, and $|a_{12}^*|^2 = |a_{21}^*|^2 =0$; for this case, relay conferencing cannot improve the AF rate.
  \item For the case with $\frac{\gamma_1}{\gamma_1 + 1} < \frac{\gamma_2}{\gamma_2 + 1 +\sigma_{21}^2}$ and $\frac{\gamma_2}{\gamma_2 + 1} < \frac{\gamma_1}{\gamma_1 + 1 +\sigma_{12}^2}$, we claim that it cannot happen. To see this point, consider the case with $C_{12} \rightarrow \infty$, i.e., $\sigma_{12}^2 \rightarrow 0$, and we obtain that $\frac{\gamma_1}{\gamma_1 + 1} < \frac{\gamma_2}{\gamma_2 + 1 +\sigma_{21}^2} < \frac{\gamma_2}{\gamma_2 + 1}$. Applying a similar argument for the other inequality, it follows that $\frac{\gamma_2}{\gamma_2 + 1} < \frac{\gamma_1}{\gamma_1 + 1 }$, which contradicts with the previous inequality. As such, it is concluded that this case cannot happen.
  \item For the case with $\frac{\gamma_1}{\gamma_1 + 1} < \frac{\gamma_2}{\gamma_2 + 1 +\sigma_{21}^2}$ and $\frac{\gamma_2}{\gamma_2 + 1} \geq \frac{\gamma_1}{\gamma_1 + 1 +\sigma_{12}^2}$, we have $|a_{11}^*|^2=|a_{12}^*|^2=0$, $|a_{22}^*|^2=\frac{1}{\gamma_2 + 1} $, and $|a_{21}^*|^2= \frac{1}{\gamma_{2} + 1 +\sigma^2_{21} }$; for this case, the source does not need to send information to relay 1, and both of the two relays forward the signals received at relay 2 to the destination.
  \item For the case with $\frac{\gamma_1}{\gamma_1 + 1} \geq \frac{\gamma_2}{\gamma_2 + 1 +\sigma_{21}^2}$ and $\frac{\gamma_2}{\gamma_2 + 1} < \frac{\gamma_1}{\gamma_1 + 1 +\sigma_{12}^2}$, we have $|a_{11}^*|^2 = \frac{1}{\gamma_1 + 1}$, $|a_{12}^*|^2=\frac{1}{\gamma_{1} + 1 +\sigma^2_{12} }$, and $|a_{22}^*|^2=|a_{21}^*|^2=0$; this case leads to results opposite to case 3).
\end{enumerate}

\begin{Remark}
From the above analysis, we conclude that for the case with $\widetilde{\gamma}_i \rightarrow 0$, after obtaining the two signals from the source and its counterpart relay, each relay should only forward the one with a higher SNR and discard the other one. Moreover, cases 3 and 4 suggest that in the low SNR regime, increasing power gain is more critical than increasing multiplex gain for the AF relaying scheme. This is opposite to the result in the high SNR regime, where the alternative relaying scheme is optimal in the sense of achieving the full multiplexing gain as stated in the introduction part.
\end{Remark}

\section{Conferencing Strategy II}

In this section, we consider conferencing strategy II, and derive the capacity upper bound and the DF achievable rate. We will show that the DF rate will be greatly improved compared with that of strategy I. As we discussed before, we omit the analysis for the AF relaying scheme, since its result is only different by some constant factors from that of strategy I.

\subsection{Capacity Upper bound}

Note that for this case, the conferencing links can be fully utilized, and thus there will be no penalty terms $\lambda_i$'s over the conferencing link rates $C_{i,3-i}$'s. As such, we obtain the following result for the capacity upper bound, which is similar to Theorem \ref{upper_bound}.

\begin{Theorem} \label{upper_bound_II}
\emph{Under conferencing strategy II, the capacity upper bound for the alternative relaying diamond channel with conferencing links is given as}
\begin{align} \label{upper_bound_number_delayII}
C_{\text{upper}} &= \max_{\lambda_1 + \lambda_2 =1 } \min \left\{  \lambda_1 \mathcal{C} (\gamma_1) +  \lambda_2 \mathcal{C} (\gamma_2), \lambda_2 \mathcal{C} (\widetilde{\gamma}_1) +  \lambda_1 \mathcal{C} (\widetilde{\gamma}_2),  \lambda_2 \left(  \mathcal{C} (\gamma_2) + \mathcal{C} (\widetilde{\gamma}_1) \right) + C_{12} +  C_{21} , \right. \nonumber \\
&~~~~~~~~~~~~~~~~~~~~~~~~~~~~~~~~\left. \lambda_1 \left(  \mathcal{C} (\gamma_1) + \mathcal{C} (\widetilde{\gamma}_2)   \right) +  C_{12}  +C_{21}  \right\}.
\end{align}
\end{Theorem}

\subsection{DF Achievable Rate}

The DF coding scheme under conferencing strategy II is similar to that for scheme I, and thus we only describe its differences from the previous scheme. For the $k$-th source message (sending to the $i$-th relay), it contains two sub-messages $w_k^i$ and $w_k^{3-i}$ via superposition coding. After relay $i$ decodes them, it sends $w_k^i$ to the destination in the $(k+1)$-th time slot, and $w_k^{3-i}$ to the other relay in the $(k+1)$-th and the $(k+2)$-th time slots via the conferencing link. As such, the rate of message $w_k^{3-i}$ is no longer subject to the conferencing link rates constraints, i.e., $R_{i,3-i} \leq  C_{i,3-i},~i=1,2$ is not required. In the $(k+3)$-th time slot, relay $3-i$ sends $w_k^{3-i}$ to the destination together with the message $w_{k+3}^{3-i}$. Accordingly, we have the following result for the DF rate.

\begin{Theorem}
\emph{Under conferencing strategy II, the DF achievable rate for the alternative relaying diamond channel with conferencing links is given as}
\begin{align}
\text{P3}:~& \max ~~R_{\text{DF}}  = R_{11} +R_{12} +R_{21} +R_{22} \label{DF_rate_II} \\
\text{s. t.}~~ & R_{3-i,i} \leq  C_{3-i,i},~i=1,2, \label{DF_rate_II_conference_constr} \\
& R_{i,i} + R_{i,3-i} \leq \lambda_i \mathcal{C} (\gamma_i),~i=1,2,  \label{DF_rate_II_constr1} \\
& R_{3-i,3-i} + R_{i,3-i} \leq \lambda_i \mathcal{C} (\widetilde{\gamma}_{3-i}) ,~i=1,2, \label{DF_rate_II_constr2} \\
& R_{i,j} \geq 0, i,j \in \{1,2\},~ \lambda_1 + \lambda_2 =1,~i=1,2,
\end{align}
\emph{where $R_{ii}$ and $R_{i,3-i}$, $i=1,2$, are defined the same as those in Problem (P1).}
\end{Theorem}

Since Problem (P3) has a similar structure as Problem (P1), it can be shown that one-side conferencing is also optimal for the DF scheme under conferencing strategy II, and the optimal solution of Problem (P3) can be obtained by a similar routine as that in Section III, which is omitted for simplicity. Here, we first have the following proposition to show a capacity result for the DF scheme.

\begin{Proposition} \label{DF_schemeII_finite_prop}
\emph{Under conferencing strategy II, the DF relaying scheme achieves the corresponding capacity upper bound with finite conferencing link rates, i.e., with $C_{12}+ C_{21}$ larger than or equal to the values summarized in Table \ref{upper_cij}.
}
\end{Proposition}
\begin{proof}
See Appendix \ref{DF_schemeII_finiterate}.
\end{proof}

\begin{Remark}
Compared to the asymptotic capacity result for conferencing strategy I given in Proposition \ref{DF_schemeI_asym}, Proposition \ref{DF_schemeII_finite_prop} guarantees that conferencing strategy II is practically feasible, and only finite conferencing link rates are necessary to achieve the capacity upper bound.
\end{Remark}

\begin{Proposition} \label{DF_schemeII_equal_I}
\emph{If the conferencing link rates are symmetric, i.e., $C_{12} = C_{21}$, the DF rate under conferencing strategy II is the same as the capacity upper bound under conferencing strategy I with arbitrary channel coefficients.}
\end{Proposition}
\begin{proof}
First, for the cases that the either one of the first two terms in the min operation of (\ref{upper_bound_number}) is the smallest among these four terms, it is easy to check that it can be achieved by setting $R_{i,3-i}=0$, $i=1,2$, in Problem (P3). On the other hand, if the third term in (\ref{upper_bound_number}) is the smallest one, it is achievable for the DF scheme due to the following fact: As a similar argument of Remark \ref{DF_schemeI_no_capa}, it can be shown that
\begin{align}
R_{\text{DF}} & \leq  R_{11} +R_{12} + 2R_{21} + R_{22} \label{scheme_II_DF_achieving1}\\
 & \leq C_{12} + \lambda_2 \left( \mathcal{C} (\gamma_2) + \mathcal{C} (\widetilde{\gamma}_1)  \right), \label{scheme_II_DF_achieving2}
\end{align}
where (\ref{scheme_II_DF_achieving2}) equals the third term in (\ref{upper_bound_number}), the equality in (\ref{scheme_II_DF_achieving1}) is achieved only when $R_{21} =0$, and the equality in (\ref{scheme_II_DF_achieving2}) is achieved only when the constraints (\ref{DF_rate_II_conference_constr})-(\ref{DF_rate_II_constr2}) achieves the equality for the case of $i=2$. Thus, define $R_{12}=C_{12}$, $R_{21} =0$, $R_{11} = \lambda_2 \mathcal{C} (\widetilde{\gamma}_1) $, and $R_{22} = \lambda_2 \mathcal{C} (\gamma_2) $. Since we assume that the third term in (\ref{upper_bound_number}) is the smallest one, it is easy to check that this solution is also feasible for the constraints in (\ref{DF_rate_II_constr1}) and (\ref{DF_rate_II_constr2}) for the case of $i=1$. As such, the third term in (\ref{upper_bound_number}) is achievable for the DF scheme. For the fourth term in (\ref{upper_bound_number}), it can be shown that it is achievable by applying a similar argument as in the previous case. Therefore, this proposition is proved.
\end{proof}

For the case that the conferencing link rates are not the same, i.e., $C_{12} \neq C_{21}$, the DF rate under conferencing strategy II may be either larger or smaller than the capacity upper bound under conferencing strategy I, which will be shown in the next section by numerical results.

\section{Numerical Results}

In this section, we present some numerical results to compare the performances of the proposed coding schemes. Here, we only consider the asymmetric channel case, i.e., $\gamma_1 = \widetilde{\gamma}_2$ and $\gamma_2 = \widetilde{\gamma}_1$, and also show the performance of the DF simultaneous relaying scheme given in \cite{chuan2} as a comparison, which usually performs the best among various coding schemes under the simultaneous relaying mode.

In Fig. \ref{asym_diff_SNR}, for the two relay conferencing strategies, we plot the capacity upper bounds and various achievable rates as functions of link gains. Here, we let $C_{12}=C_{21} =5$ bits/s/Hz, $\gamma_1 = \widetilde{\gamma}_2 =10$ dB, and $\gamma_2 = \widetilde{\gamma}_1$ change over $[-10,30]$ dB. It is observed that the two upper bounds coincide when the channel gain is relatively small, i.e., when below 15 dB. For the DF relaying scheme, it achieves the capacity upper bound when the channel gain is less than 10 dB under conferencing strategy II, while only at the point of 10 dB under conferencing strategy I.

In Fig. \ref{conference_rate}, we plot the capacity upper bounds and various achievable rates as functions of the conferencing link rates for both of the two conferencing strategies. Here, we assume $C_{12} = C_{21}$, $\gamma_1 = \widetilde{\gamma}_2 = 10$ dB, and $\gamma_2 = \widetilde{\gamma}_1 =30$ dB. It is observed that relay conferencing can significantly increase these achievable rates for both of the simultaneous and alternative relaying schemes. Moreover, although it is proved in Proposition \ref{DF_schemeI_asym} that under conferencing strategy I, the alternative DF scheme can asymptotically achieve the capacity upper bound as $C_{12}$ goes to infinity, unfortunately it approaches the upper bound very slowly: Even when $C_{12}$ are 50 bits/s/Hz, the gap between them is still about 1 bits/s/Hz; while under conferencing strategy II, the DF scheme achieves the corresponding upper bound with relative small conferencing link rates, i.e., about 12 bits/s/Hz.

\section{Conclusion}
In this paper, we considered the alternative relaying diamond relay channel with conferencing links. We derived the DF and AF achievable rates for two conferencing strategies, and showed that these rate maximization problem are convex. For the DF relaying scheme, by further exploiting the properties of the optimal solution, one-side conferencing was shown to be optimal for the DF scheme with both of the two conferencing strategies. Then, we obtained the DF rate in closed-form, and explicitly showed the rules on which conferencing link should be used under given channel conditions for one-side conferencing. Interestingly, the DF scheme was shown to be capacity-achieving with the help of finite conferencing link rates under conferencing strategy II, whose lower bounds were also derived. For the AF relaying scheme, we studied the optimal combining strategy, and showed that one-side conferencing is not optimal in general. Furthermore, some asymptotic optimal combining strategies were obtained in both the high and low SNR regimes.

%
\IEEEpeerreviewmaketitle

\appendices
\section{Achievability Proof of the DF Rate} \label{appen_alternative_DF}

\textbf{Codebook Generation:} First, note that we only need to generate in total two sets of codebooks for the odd and even time slots, respectively; for simplicity, we use the subscript $s$, $s=1,2$, to distinguish these two sets of codebooks, i.e., $s=1$ for that in odd time slots and $s=2$ for that in even ones. The codebooks at the source are generated as follows: Generate $2^{n R_{s,s}}$ i.i.d. sequences $\bm{u}_s(q^{1}_s)$, where $q^{1}_s \in \left[ 1:  2^{nR_{s,s}} \right]$, according to the distribution $\prod_{j=1}^{\lambda_s n} p(u_{s,j})$; for each $\bm{u}_s( q^{1}_s )$, generate $2^{n R_{s,3-s}}$ i.i.d. sequences $\bm{x}_s (q^{1}_s,q^{2}_s)$, where $q^{2}_s \in \left[ 1: 2^{n R_{s,3-s}}\right]$, with the distribution $\prod_{j=1}^{\lambda_s n} p(x_{s,j} | u_{s,j})$. For the conferencing links, generate $2^{n R_{3-s,s}}$ i.i.d. sequences $\mathcal{M}_{3-s}(v_{3-s})$, where $v_{3-s} \in \left[1:R_{3-s,s} \right]$. For the relay-destination transmissions, generate $2^{n \left(R_{3-s,3-s} +  R_{s,3-s} \right) }$ i.i.d. sequences $\bm{t}_{3-s} (s_{s}^1,s_s^2)$ by a similar superposition coding method as that of $\bm{x}_s$, where $s_{s}^1 \in \left[ 1:R_{3-s,3-s}\right]$ and $s_{s}^2 \in \left[ 1: R_{s,3-s}\right]$, according to the distribution $\prod_{j=1}^{\lambda_s n} p(t_j)$.

\textbf{Encoding and decoding:} At the beginning of the $k$-th time slot, $k=1,2,\cdots$, where the source sends message to the $i$-th relay with $i=1$ for odd $k$ and $i=2$ for even $k$, the source splits the message $w_k$ into two submessages $w^{1}_k$ and $w^{2}_k$, and transmits $\bm{x}_i (w^{1}_k, w^{2}_k)$; the $(3-i)$-th relay transmits $\mathcal{M}_i (w^{3-i}_{k-1})$ to the $i$-th relay via the conferencing link; the $(3-i)$-th relay transmits $\bm{t}_i (w^{3-i}_{k-1},w^{i}_{k-2})$ to the destination.

At the end of the $k$-th time slot, the $i$-th relay obtains $\mathcal{M}_i (w^{3-i}_{k-1})$ from the $(3-i)$-th relay. Since we assume that the conferencing links are noiseless, the $i$-th relay can successfully decode message $w^{3-i}_{k-1}$ if
\begin{align}
R_{3-i,i} \leq \lambda_i C_{3-i,i}.
\end{align}
Simultaneously, the $i$-th relay obtains $\bm{y}_i$ from the source. Then, it decodes $(w^{1}_k, w^{2}_k)$, and this can be done reliably if
\begin{align}
R_{i,i} + R_{i,3-i}  \leq \lambda_i I \left( X_i ; Y_i  \right)  =  \lambda_i \mathcal{C} \left( \gamma_i \right), \label{DF_first_hop}
\end{align}
where (\ref{DF_first_hop}) is obtained by choosing $X_i$ as a Gaussian random variable with a distribution $\mathcal{CN} \left( 0, P_S \right)$. At the destination, it decodes $(w^{3-i}_{k-1},w^{i}_{k-2})$, and this can be done reliably if
\begin{align}
R_{3-i,3-i} +  R_{i,3-i} \leq \lambda_i I \left( T_{3-i} ; Z_i \right)  = \lambda_i \mathcal{C} \left( \widetilde{\gamma}_{3-i} \right), \label{DF_second_hop}
\end{align}
where (\ref{DF_second_hop}) is obtained by choosing $T_{3-i}$ as a Gaussian random variable with a distribution $\mathcal{CN} \left( 0, P_R \right)$.
Based on the above analysis, we can obtain the DF achievable rate as shown in (\ref{DF_rate}).

\section{Proof of Proposition \ref{af_convex_prop}} \label{convex_af}
To prove this result, we only need to show that $R_i$ defined in (\ref{AF_rate_original}) is concave over $|a_{i,j}|^2$'s \cite{boyd}, due to the convexity of the constraints in (\ref{AF_power_constraints}). Then, since the function $y=\mathcal{C}(x)$ is concave and non-decreasing, we need to prove that the function within the $\mathcal{C}(\cdot)$ function in (\ref{AF_rate_original}) is concave \cite{boyd}. Moreover, noticing that $\frac{ |a_{ii}|^2 \gamma_i \widetilde{\gamma}_i}{1  +  |a_{ii}|^2 \widetilde{\gamma}_i}$ is concave over $|a_{ii}|^2$, we only need to show that the second fraction in $\mathcal{C}(\cdot)$ is also concave. By letting $x=|a_{i,3-i}|^2$ and $y=|a_{3-i,3-i}|^2$, and normalizing the coefficients of $x$ in the numerator and the denominator both to 1, it is equivalent to prove that $z =  \frac{ x }{ x + a y + b }$, where $a$ and $b$ are some positive constants, is concave. Then, check the Hessian matrix of function $z$ as
\begin{align}
\bm{H} = \frac{1}{(x+ay+b)^3} \left[ \begin{array}{cc}
         -2ay-2b & ax-a^2y-ab \\
         ax-a^2y-ab & 2a^2x
       \end{array}
 \right].
\end{align}
Noticing that $a>0,~b>0,~x\geq 0,$ and $y\geq 0$, it is easy to show that $-2ay-2b <0$ and $|\bm{H}| = \frac{1}{(x+ay+b)^3} \left[ - 2a^2x ( 2ay+ 2b) - (ax-a^2y-ab)^2 \right] < 0$, which implies that $\bm{H}$ is negative semidefinite and function $z$ is concave. Therefore, the proposition is proved.

\section{Proof of Proposition \ref{DF_schemeII_finite_prop}} \label{DF_schemeII_finiterate}
Similar to the proof of Proposition \ref{DF_schemeII_equal_I}, it is easy to check that the first and the second terms in (\ref{upper_bound_number_delayII}) can be achieved by the DF rate given in Problem (P3). However, for the third and the fourth terms in (\ref{upper_bound_number_delayII}), by a similar argument as Remark \ref{DF_schemeI_no_capa}, it can be shown that these two terms cannot be achieved by the DF relaying scheme. As such, the DF relaying scheme achieves the capacity upper bound only for the case that the last two terms are redundant, i.e., for the optimal $\lambda_i^*$ achieving the maximum value of the following optimization problem,
\begin{align} \label{upper_schemeII_sim}
\widetilde{C}_{\text{upper}} &= \max_{\lambda_1 + \lambda_2 =1 } \min \left\{  \lambda_1 \mathcal{C} (\gamma_1) +  \lambda_2 \mathcal{C} (\gamma_2), \lambda_2 \mathcal{C} (\widetilde{\gamma}_1) +  \lambda_1 \mathcal{C} (\widetilde{\gamma}_2) \right\},
\end{align}
we always have $ \lambda_2^* \left( \mathcal{C} (\gamma_2) + \mathcal{C} (\widetilde{\gamma}_1)\right) + C_{12} + C_{21} \geq \widetilde{C}_{\text{upper}}$ and $ \lambda_1^*  \left( \mathcal{C} (\gamma_1) + \mathcal{C} (\widetilde{\gamma}_2)\right) + C_{12} + C_{21} \geq \widetilde{C}_{\text{upper}}$. Therefore, the capacity upper bound is achieved by the DF scheme only when the following relationship is satisfied
\begin{align}
C_{12} + C_{21} \geq \widetilde{C}_{\text{upper}} - \min \left\{ \lambda_2^* \left( \mathcal{C} (\gamma_2) + \mathcal{C} (\widetilde{\gamma}_1)\right), \lambda_1^*  \left( \mathcal{C} (\gamma_1) + \mathcal{C} (\widetilde{\gamma}_2)\right) \right\}.
\end{align}
Denote
\begin{align} \label{g_func}
 g(\lambda_1^*) = \min \left\{ \lambda_2^* \left( \mathcal{C} (\gamma_2) + \mathcal{C} (\widetilde{\gamma}_1)\right), \lambda_1^*  \left( \mathcal{C} (\gamma_1) + \mathcal{C} (\widetilde{\gamma}_2)\right) \right\},
\end{align}
and it follows that $g(0)=g(1) = 0$. Then, in order to compute the lower bound on $C_{12} + C_{21}$ to achieve the capacity upper bound, we only need to compute $\widetilde{C}_{\text{upper}}$ and the corresponding $\lambda_i^*$. For Problem (\ref{upper_schemeII_sim}), it follows that
\begin{enumerate}
  \item $\gamma_1 > \gamma_2, \widetilde{\gamma}_2 > \widetilde{\gamma}_1$: It is obtained that $\lambda^*_1=1$, and thus, $\widetilde{C}_{\text{upper}} = \min \left\{ \mathcal{C}(\gamma_1), \mathcal{C} (\widetilde{\gamma}_2)  \right\}$;
  \item $\gamma_1 \leq \gamma_2, \widetilde{\gamma}_2 \leq \widetilde{\gamma}_1$: It is obtained that $\lambda^*_1 =0$, and thus $\widetilde{C}_{\text{upper}} = \min \left\{ \mathcal{C}(\gamma_2) , \mathcal{C}(\widetilde{\gamma}_1)  \right\}$;
  \item $\gamma_1 > \gamma_2 ,  \widetilde{\gamma}_1 \leq \widetilde{\gamma}_2 , \widetilde{\gamma}_2 \geq \gamma_1$: It is obtained that $\lambda^*_1 = 1$, and thus $\widetilde{C}_{\text{upper}} = \mathcal{C}(\widetilde{\gamma}_2)$;
  \item $\gamma_1 \leq  \gamma_2 ,  \widetilde{\gamma}_1 > \widetilde{\gamma}_2 , \widetilde{\gamma}_2 \leq \gamma_1$: It is obtained that $\lambda^*_1 = 1$, and thus $\widetilde{C}_{\text{upper}} = \mathcal{C} \left( \gamma_1 \right)$;
  \item $\gamma_1 >  \gamma_2 ,  \widetilde{\gamma}_1 > \widetilde{\gamma}_2 , \gamma_1 > \widetilde{\gamma}_2$: It is obtained that $\lambda^*_1 = \lambda_0$, and thus $\widetilde{C}_{\text{upper}} = \frac{\mathcal{C} \left( \gamma_2 \right) \mathcal{C} \left( \widetilde{\gamma}_2 \right) - \mathcal{C} \left( \gamma_1 \right) \mathcal{C} \left( \widetilde{\gamma}_1 \right)}{ \mathcal{C} \left( \widetilde{\gamma}_2 \right) -  \mathcal{C} \left( \widetilde{\gamma}_1 \right) - \mathcal{C} \left( \gamma_1 \right)+ \mathcal{C} \left( \gamma_2 \right)}$;
  \item $\gamma_1 \leq  \gamma_2 ,  \widetilde{\gamma}_1 < \widetilde{\gamma}_2 , \gamma_1 <  \widetilde{\gamma}_2 $: It is obtained that $\lambda^*_1 = \lambda_0$, and thus $\widetilde{C}_{\text{upper}} = \frac{\mathcal{C} \left( \gamma_2 \right) \mathcal{C} \left( \widetilde{\gamma}_2 \right) - \mathcal{C} \left( \gamma_1 \right) \mathcal{C} \left( \widetilde{\gamma}_1 \right)}{ \mathcal{C} \left( \widetilde{\gamma}_2 \right) -  \mathcal{C} \left( \widetilde{\gamma}_1 \right) - \mathcal{C} \left( \gamma_1 \right)+ \mathcal{C} \left( \gamma_2 \right)}$;
\end{enumerate}
where $\lambda_0 = \frac{\mathcal{C} \left( \gamma_2 \right) - \mathcal{C} \left( \widetilde{\gamma}_1 \right)}{ \mathcal{C} \left( \widetilde{\gamma}_2 \right) -  \mathcal{C} \left( \widetilde{\gamma}_1 \right) - \mathcal{C} \left( \gamma_1 \right)+ \mathcal{C} \left( \gamma_2 \right)} $.
Thus, the lower bound on $C_{12} + C_{21}$ to achieve the capacity upper bound is given in Table \ref{upper_cij}.



\begin{figure}[h]
\centering
\includegraphics[width=.7\linewidth]{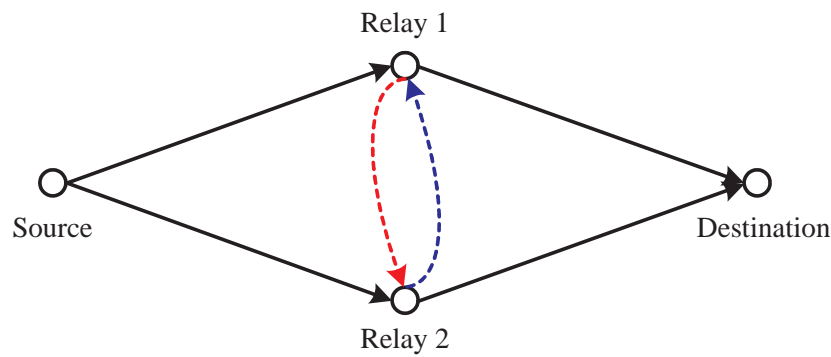}
\caption{The diamond relay channel with out-of-band conferencing links.} \label{system_model}
\end{figure}

\begin{figure}[h]
\centering
\includegraphics[width=.7\linewidth]{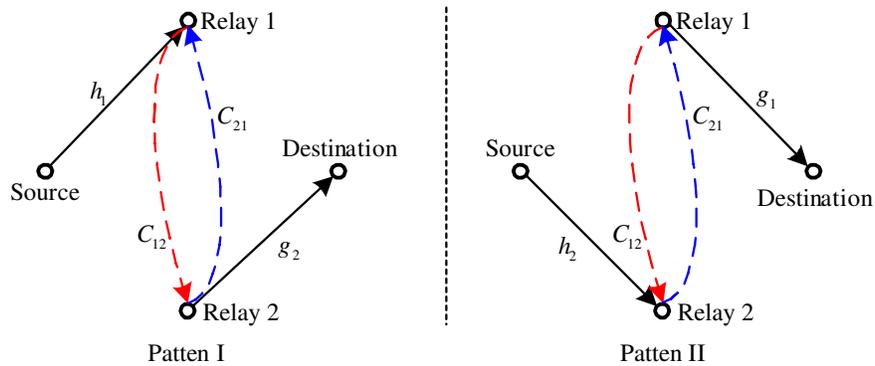}
\caption{Two transmitting-receiving states for the alternative relaying Diamond channel with conferencing links.} \label{alternative_relay}
\end{figure}

\begin{figure}[h]
\centering
\subfigure[Conferencing strategy I.]{
          \label{fig21}
          \includegraphics[width=.7 \linewidth]{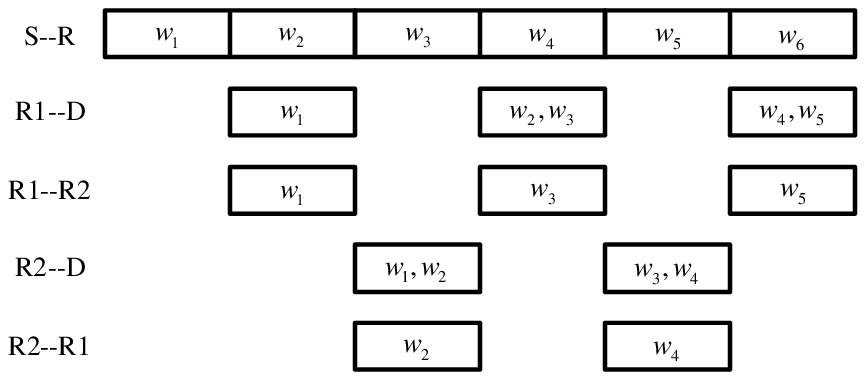}}
\hspace{.45in}
\subfigure[Conferencing strategy II.]{
          \label{fig22}
          \includegraphics[width=.7 \linewidth]{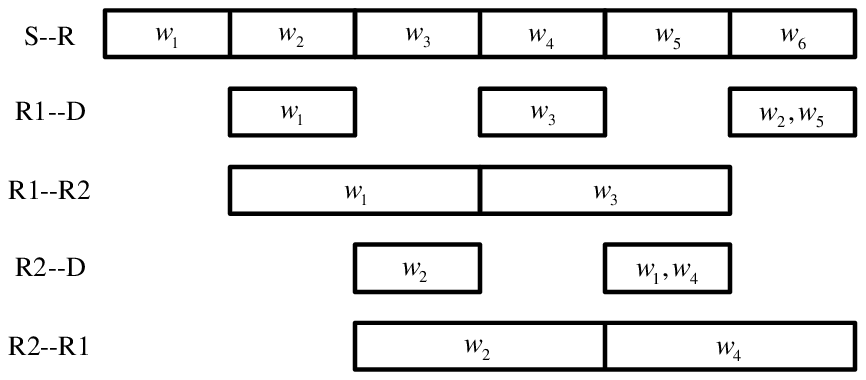}}
\caption{Relay conferencing strategies for the diamond relay channel with conferencing links, with ``S'', ``R1'', ``R2'', and ``D'' denoting the source, relay 1, relay 2, and destination, respectively.}
\label{fig_scheduling}
\end{figure}

\begin{figure}[h]
\centering
\subfigure[Case 1)]{
          \label{four_cases1}
          \includegraphics[width=.4 \linewidth]{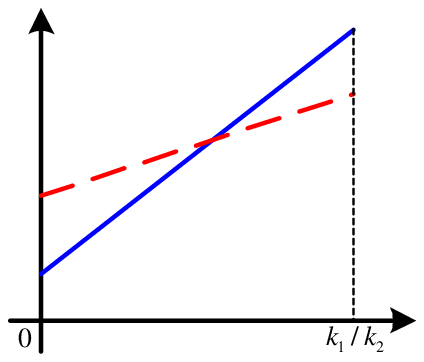}}
\hspace{.45in}
\subfigure[Case 2)]{
          \label{four_cases2}
          \includegraphics[width=.4 \linewidth]{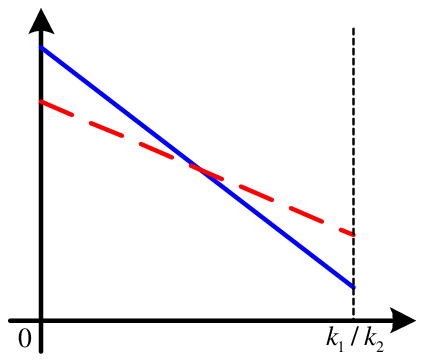}}
\hspace{.45in}
\subfigure[Case 3)]{
          \label{four_cases3}
          \includegraphics[width=.4 \linewidth]{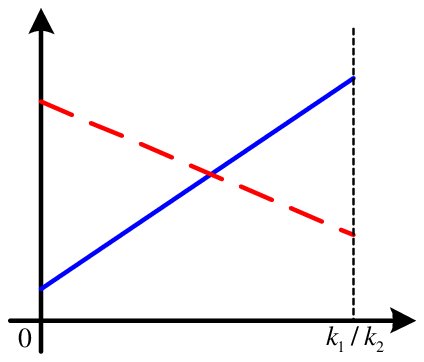}}
\hspace{.45in}
\subfigure[Case 4)]{
          \label{four_cases4}
          \includegraphics[width=.4 \linewidth]{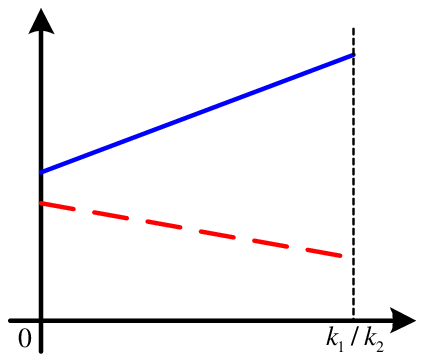}}
\caption{Four possible cases for the max/min solutions in (\ref{DF_casei}) and (\ref{DF_caseii}), where $k_1$ and $k_2$ are given in (\ref{constraint_r12}) and (\ref{constraint_r21}), respectively.}
\label{four_cases}
\end{figure}

\begin{table}
\centering
\begin{threeparttable}[t]
\caption{The DF rate under the conferencing strategy I and the corresponding one-side conferencing scheme.}
\begin{tabular}{ c || c |c }
    \hline \hline
    Channel conditions &  DF rate & Conferencing scheme \\ \hline \hline
    $\gamma_1 > \gamma_2$ and $\widetilde{\gamma}_2 > \widetilde{\gamma}_1$ & (\ref{DF_rate_case11}) & \\
     $\gamma_1 > \gamma_2$, $\widetilde{\gamma}_2 \leq \widetilde{\gamma}_1$ and $\mathcal{C} (\gamma_1) \mathcal{C} (\gamma_2) < \mathcal{C} (\widetilde{\gamma}_1) \mathcal{C} (\widetilde{\gamma}_2)$ & (\ref{DF_inter}) &  Relay 1 $\rightarrow$ Relay 2\tnote{1}  \\
     $\gamma_1 \leq \gamma_2$, $\widetilde{\gamma}_2 > \widetilde{\gamma}_1$ and $\mathcal{C} (\gamma_1) \mathcal{C} (\gamma_2) > \mathcal{C} (\widetilde{\gamma}_1) \mathcal{C} (\widetilde{\gamma}_2)$  & (\ref{DF_inter}) &  \\ \hline
     $\gamma_1 < \gamma_2$ and $\widetilde{\gamma}_2 < \widetilde{\gamma}_1$ &  (\ref{DF_rate_case22}) & \\
 $\gamma_1 > \gamma_2$, $\widetilde{\gamma}_2 \leq \widetilde{\gamma}_1$ and $\mathcal{C} (\gamma_1) \mathcal{C} (\gamma_2) > \mathcal{C} (\widetilde{\gamma}_1) \mathcal{C} (\widetilde{\gamma}_2)$ & (\ref{DF_rate_case21}) & Relay 2 $\rightarrow$ Relay 1 \\
     $\gamma_1 \leq \gamma_2$, $\widetilde{\gamma}_2 > \widetilde{\gamma}_1$ and $\mathcal{C} (\gamma_1) \mathcal{C} (\gamma_2) < \mathcal{C} (\widetilde{\gamma}_1) \mathcal{C} (\widetilde{\gamma}_2)$ & (\ref{DF_rate_case21}) & \\ \hline
     $\gamma_1 = \gamma_2$ and $\widetilde{\gamma}_2 = \widetilde{\gamma}_1$ & (\ref{DF_rate_0})  & \\
     $\gamma_1 > \gamma_2$, $\widetilde{\gamma}_2 \leq \widetilde{\gamma}_1$ and $\mathcal{C} (\gamma_1) \mathcal{C} (\gamma_2) = \mathcal{C} (\widetilde{\gamma}_1) \mathcal{C} (\widetilde{\gamma}_2)$& (\ref{DF_rate_0}) &  No relay conferencing  \\
     $\gamma_1 \leq \gamma_2$, $\widetilde{\gamma}_2 > \widetilde{\gamma}_1$ and $\mathcal{C} (\gamma_1) \mathcal{C} (\gamma_2) = \mathcal{C} (\widetilde{\gamma}_1) \mathcal{C} (\widetilde{\gamma}_2)$ & (\ref{DF_rate_0}) & \\
    \hline \hline
\end{tabular}
\label{which_conf_link}
\begin{tablenotes}
\item 1. This means that only the conferencing link from relay 1 to relay 2 is used, and the other one is not.
\end{tablenotes}
\end{threeparttable}
\end{table}

\begin{table}
\centering
\begin{threeparttable}[t]
\caption{Lower bound on $C_{12}+C_{21}$ for the DF scheme to achieve the capacity upper bound under conferencing strategy II.}
\begin{tabular}{ c || c  }
    \hline \hline
    Channel conditions & Minimum $C_{12}+C_{21}$  \\ \hline \hline
    $\gamma_1 > \gamma_2, \widetilde{\gamma}_2 > \widetilde{\gamma}_1$ & $ \min \left\{ \mathcal{C}(\gamma_1), \mathcal{C} (\widetilde{\gamma}_2)  \right\}$ \\ \hline
     $\gamma_1 \leq \gamma_2, \widetilde{\gamma}_2 \leq \widetilde{\gamma}_1$ & $ \min \left\{ \mathcal{C}(\gamma_2) , \mathcal{C}(\widetilde{\gamma}_1)  \right\}$ \\ \hline
     $\gamma_1 > \gamma_2 ,  \widetilde{\gamma}_1 \leq \widetilde{\gamma}_2 , \widetilde{\gamma}_2 \geq \gamma_1$  &  $\mathcal{C}(\widetilde{\gamma}_2)$ \\ \hline
     $\gamma_1 \leq  \gamma_2 ,  \widetilde{\gamma}_1 > \widetilde{\gamma}_2 , \widetilde{\gamma}_2 \leq \gamma_1$  & $\mathcal{C} \left( \gamma_1 \right)$ \\ \hline
$\gamma_1 >  \gamma_2 ,  \widetilde{\gamma}_1 > \widetilde{\gamma}_2 , \gamma_1 > \widetilde{\gamma}_2$  & $\frac{\mathcal{C} \left( \gamma_2 \right) \mathcal{C} \left( \widetilde{\gamma_2} \right) - \mathcal{C} \left( \gamma_1 \right) \mathcal{C} \left( \widetilde{\gamma_1} \right)}{ \mathcal{C} \left( \widetilde{\gamma_2} \right) -  \mathcal{C} \left( \widetilde{\gamma_1} \right) - \mathcal{C} \left( \gamma_1 \right)+ \mathcal{C} \left( \gamma_2 \right)} - g(\lambda_0)$\tnote{1}  \\ \hline
     $\gamma_1 \leq  \gamma_2 ,  \widetilde{\gamma}_1 < \widetilde{\gamma}_2 , \gamma_1 <  \widetilde{\gamma}_2 $  &  $\frac{\mathcal{C} \left( \gamma_2 \right) \mathcal{C} \left( \widetilde{\gamma_2} \right) - \mathcal{C} \left( \gamma_1 \right) \mathcal{C} \left( \widetilde{\gamma_1} \right)}{ \mathcal{C} \left( \widetilde{\gamma_2} \right) -  \mathcal{C} \left( \widetilde{\gamma_1} \right) - \mathcal{C} \left( \gamma_1 \right)+ \mathcal{C} \left( \gamma_2 \right)} - g(\lambda_0)$ \\
    \hline \hline
\end{tabular}
\label{upper_cij}
\begin{tablenotes}
\item 1. $g(\cdot)$ is defined in (\ref{g_func}), and $\lambda_0 = \frac{\mathcal{C} \left( \gamma_2 \right) - \mathcal{C} \left( \widetilde{\gamma}_1 \right)}{ \mathcal{C} \left( \widetilde{\gamma}_2 \right) -  \mathcal{C} \left( \widetilde{\gamma}_1 \right) - \mathcal{C} \left( \gamma_1 \right)+ \mathcal{C} \left( \gamma_2 \right)} $.
\end{tablenotes}
\end{threeparttable}
\end{table}

\begin{figure}[h]
\centering
\includegraphics[width=.6\linewidth]{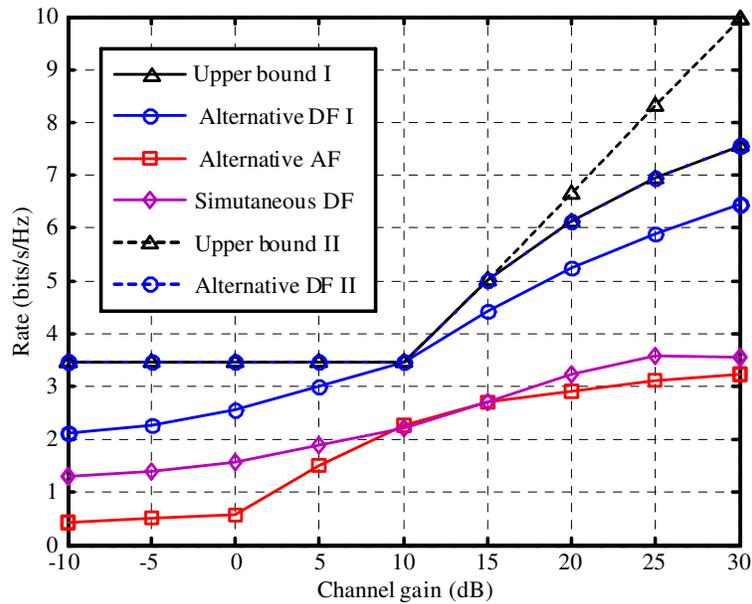}
\caption{Comparison of the capacity upper bounds and various achievable rates under different channel conditions, with $C_{12}=C_{21} =5$ bits/s/Hz, $\gamma_1 = \widetilde{\gamma}_2 = 10$ dB, and different $\gamma_2 = \widetilde{\gamma}_1$.} \label{asym_diff_SNR}
\end{figure}

\begin{figure}[h]
\centering
\includegraphics[width=.6\linewidth]{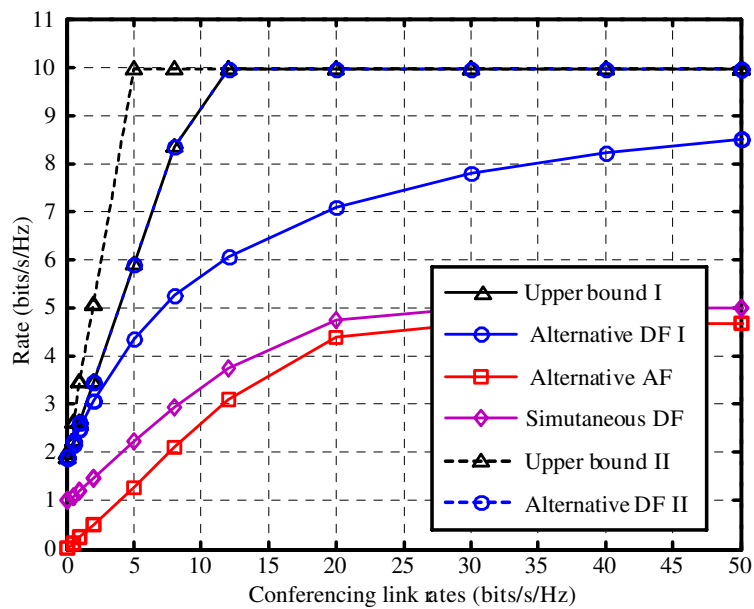}
\caption{Comparison of the capacity upper bounds and various achievable rates over different conferencing link rates, with $C_{12} =C_{21}$, $\gamma_1 = \widetilde{\gamma}_2 = 10$ dB, and $\gamma_2 = \widetilde{\gamma}_1 =30$ dB.} \label{conference_rate}
\end{figure}

%
%
%



%

\end{document}